\documentclass[aps,%
twocolumn,%
floatfix,%
pra,%
amsfonts,%
groupedaddress,%
superscriptaddress]{revtex4-2}%

\usepackage{amsmath}
\usepackage{amssymb}
\usepackage{amsthm}
\usepackage{bm,bbm} 
\usepackage[table,dvipsnames,usenames]{xcolor}
\usepackage{graphicx}
\usepackage{indentfirst}
\usepackage{enumitem}
\usepackage{tikz}
\usepackage{quantikz}
\usetikzlibrary{graphs} 
\usepackage[titletoc,toc,page,title]{appendix}
\definecolor{beamer@blendedblue}{rgb}{0.2,0.2,0.7}      
\usepackage[bookmarks=false]{hyperref}
\hypersetup{colorlinks=true,%
citecolor=beamer@blendedblue,%
linkcolor=beamer@blendedblue,%
urlcolor=beamer@blendedblue,%
pdfstartview=FitH,%
bookmarksopen=true}

\usepackage[T1]{fontenc}
\usepackage{times}
\usepackage{booktabs} 
\newtheorem{theorem}{Theorem}
\newtheorem{proposition}[theorem]{Proposition}
\newtheorem{lemma}{Lemma}

\usepackage{pifont}
\newlist{todolist}{itemize}{2}
\setlist[todolist]{label=$\square$}
\DeclareMathOperator{\tr}{Tr}

\newcommand{\ox}{\otimes}
\newcommand{\1}{\mathbbm{1}}

\newcommand*{\cF}{\mathcal{F}}
\newcommand*{\cH}{\mathcal{H}}
\newcommand*{\cO}{\mathcal{O}}
\newcommand*{\cP}{\mathcal{P}}
\newcommand*{\cS}{\mathcal{S}}
\newcommand*{\bF}{\mathbb{F}}
\newcommand*{\bP}{\mathbb{P}}

\definecolor{alizarin}{rgb}{0.82, 0.1, 0.26}
\definecolor{googleblue}{HTML}{4285F4}
\definecolor{googlered}{HTML}{DB4437}
\definecolor{googleyellow}{HTML}{F4B400}
\definecolor{googlegreen}{HTML}{0F9D58}
\definecolor{klevinblue}{HTML}{002FA7}
\definecolor{tiffanyblue}{HTML}{0ABAB5}

\newcommand\prlsection[1]{\textit{\textbf{#1}}---}
\newcommand\prlsubsection[1]{\paragraph*{#1\!\!\!\!}}

\begin{document}

\title{Quantum memory assisted entangled state verification with local measurements}
\author{Siyuan Chen}
\affiliation{Hefei National Research Center for Physical Sciences at the Microscale and School of Physical Sciences, University of Science and Technology of China, Hefei 230026, China}%
\affiliation{College of Physics, Jilin University, Changchun 130012, China}%
\author{Wei Xie}%
\affiliation{School of Computer Science and Technology,%
University of Science and Technology of China, Hefei 230027, China}%
\author{Ping Xu}
\affiliation{Institute for Quantum Information \& State Key Laboratory of High Performance Computing, %
College of Computer Science and Technology, National University of Defense Technology, %
Changsha 410073, China}
\author{Kun Wang}
\email{nju.wangkun@gmail.com}%
\affiliation{Institute for Quantum Information \& State Key Laboratory of High Performance Computing, %
College of Computer Science and Technology, National University of Defense Technology, %
Changsha 410073, China}

\begin{abstract}
\noindent 
We consider the quantum memory assisted quantum state verification task,
where an adversary prepare independent multipartite entangled states and send to the local verifiers, 
who then store several copies in the quantum memory and measure them collectively to make decision.
We establish an exact analytic formula for optimizing two-copy state verification, where the verifiers store two copies, 
and give a globally optimal two-copy strategy for 
multi-qubit graph states involving only Bell measurements.
When the verifiers can store arbitrarily many copies, 
we present a dimension expansion technique that designs efficient verification strategies for this case, 
showcasing its application to efficiently verifying GHZ-like states. 
These strategies become increasingly advantageous with growing memory resources, 
ultimately approaching the theoretical limit of efficiency.
Our findings demonstrate that quantum memories enhance state verification efficiency,
sheding light on error-resistant strategies and practical applications 
of large-scale quantum memory-assisted verification.
\end{abstract}
\date{\today}
\maketitle

\prlsection{Introduction.}The precise and efficient characterization of quantum states 
is a pivotal endeavor 
in many quantum information processing tasks such as quantum teleportation~\cite{PhysRevLett.70.1895}, 
quantum cryptography~\cite{pirandola2020advances}, 
and measurement-based quantum computation~\cite{Gottesman_1999}. 
While the tomography method theoretically possesses the capability to reconstruct the complete density matrix~\cite{D_Ariano_2002,PhysRevLett.129.133601,PhysRevA.105.032427}, its computational demands and time-consuming nature become 
particularly pronounced as the size of the quantum system increases, due to the curse of dimensionality. 
Fortunately, the need for tomography diminishes when our focus is narrowed to 
specific characteristics of quantum systems. 
Numerous statistical methods have been devised for 
quantum certification, validation, and benchmarking~\cite{eisert2020quantum,kliesch2021theory}. 
Among these, quantum state verification (QSV)~\cite{Hayashi2015Verifiable,pallister2018optimal,Zhu_2019_adver,Wang_2019,Li_2019,Yu_2019,Liu_2021,miguel2022collective,Zhang_2020,jiang2020towards}
not only accurately estimates the quality of the quantum states but also consumes an exponentially 
smaller number of quantum state copies, thus emerging as a highly potential tool.
We refer the interested readers to~\cite{Yu_2022} and references therein.

In QSV, we consider a quantum device designed to produce a multipartite pure state $\ket{\psi}$. 
Throughout this work, we assume that $\ket{\psi}$ is $n$-partite 
and each party is $d$-dimensional, with associated Hilbert space $\cH$.
However, it might work incorrectly and outputs 
independent states $\sigma_1, \sigma_2, \ldots, \sigma_N$ in $N$ runs. 
It is guaranteed that either $\sigma_j=\proj{\psi}$ for all $j$ (good case)
or $\bra{\psi}\sigma_j\ket{\psi}\leq 1-\varepsilon$ for all $j$ (bad case).
After recieving these states, a verifier performs two-outcome measurements randomly chosen 
from a set of available measurements.
Each two-outcome measurement $\{T_\ell,\1-T_\ell\}$ is specified by some operator $T_\ell$
and is performed with probability $p_\ell$, corresponding to passing the test.
In the bad case, the maximal probability that $\sigma_j$ passes the 
test satisfies~\cite{pallister2018optimal}
\begin{align}\label{eq:pallister2018optimal}
  \max_{\bra{\psi}\sigma_j\ket{\psi}\leq1-\varepsilon}\tr[\Omega\sigma_j]
= 1 - (1-\lambda_2(\Omega))\varepsilon,
\end{align}
where $\Omega=\sum_\ell p_\ell T_\ell$ is a verification strategy
and $\lambda_2(\Omega)$ is the second largest eigenvalue of $\Omega$.
In the bad case, all the $N$ sampled quantum states can pass the test 
with probability at most $[1 - (1-\lambda_2(\Omega))\varepsilon]^N$.
Hence to achieve certain fixed worst-case failure probability $\delta$, it suffices to take
\begin{align}\label{eq:onecpN}
    N(\Omega) = \frac{\ln\delta}{\ln[1 - (1-\lambda_2(\Omega))\varepsilon]}
\approx \frac{1}{(1-\lambda_2(\Omega))\varepsilon} \ln \frac{1}{\delta},
\end{align}
where $\ln$ denotes the natural logarithm and 
the approximation holds when $\varepsilon$ is small. 
We call $N(\Omega)$ the sample complexity 
of the verification strategy $\Omega$ in abuse of notation.
Specially, the globally optimal strategy $\{\proj{\psi},\1-\proj{\psi}\}$
has sample complexity $N_{\rm glob}\approx1/\varepsilon\ln1/\delta$. 
We say strategy $\Omega$ achieves globally optimal efficiency if, when $\varepsilon$ is small enough, $N(\Omega)$ has the same asymptotic behavior as $N_{\rm glob}$. 
While the globally optimal strategy offers exceptional efficiency, 
its reliance on entangled measurements poses challenges in experimental implementation. 
Thus, we focus on designing efficient strategies that leverage only local measurements 
and classical communication, making them amenable to practical applications.

\prlsection{Motivation.}Quantum memories, analogous to the digital memory used in classical computers, 
have been realized in diverse physical systems~\cite{Heshami_2016}.
For example, Bhaskar \emph{et al.}~\cite{Bhaskar_2020} demonstrated an integrated single 
solid-state spin memory for implementing asynchronous photonic Bell-state measurements, 
a crucial element in quantum repeaters. 
Advances in quantum memories offer substantial benefits to burgeoning quantum technologies 
such as quantum key distribution~\cite{Bhaskar_2020} and quantum control~\cite{roget2020quantum},
and fundamentally revolutionize our understanding of 
physical phenomena like the uncertainty principle~\cite{coles2017entropic}.
Given these promising developments, the question naturally arises: 
Can we harness quantum memories to enhance quantum state verification?

In this Letter, we first propose the $(n,d,k)$-verification strategy.
By fixing the copy number $k$, we quantitatively limit the capability of quantum memories for the verifiers.
In the case of $k = 2$, verifiers have the minimum quantum memory.
For such a two-copy verification strategy $\Omega$, we delineate the intrinsic value $\lambda_{\star}(\Omega)$ and $\varepsilon_{\max}\left(\Omega\right)$ that underpin its verification efficiency.
Using these intrinsic values, we find a globally optimal $(n,d,2)$-strategy for graph states with only one measurement setting.
Furthermore, we propose a dimension expansion method for the case of $k > 2$.
We show that for all GHZ-like states, a group of local strategies exists with $k = \{1, 2, 3, \ldots\}$
such that efficiency finally approaches global optimality as $k \to \infty$.
This work takes a crucial step in demonstrating how to construct the most efficient QSV strategy under the constraints of limited quantum memories.

Several preceding studies have demonstrated the potential of quantum memory to improve the efficiency of quantum state verification, 
albeit from different viewpoints and were applied in different ways. 
For example, Liu~\emph{et al.}~\cite{Liu_2021} constructed a universally optimal protocol for verifying entangled states by employing quantum nondemolition measurements. 
However, this protocol's practicality is limited because ancilla qubits have to be transported between different parties.
Miguel-Ramiro \emph{et al.}~\cite{miguel2022collective} introduced collective strategies for the efficient, local verification of ensembles of Bell pairs. 
However, their strategies are limited to Bell states and GHZ states with Werner-type noise and require complex error number gates (ENG).
We provide a detailed comparison of our work with these works with illustrative examples in Appendix~\ref{appx:comparision}, 
highlighting the esential differences.

\prlsection{Quantum memory assisted state verification.}In this verification strategy,
$n$ spatially disparate verifiers conduct a test as follows:
First, they store $k$ copies of $d$-dimensional qudits in their local quantum memories;
Then, they measure their local copies in $\cH^k\equiv\cH^{\otimes k}$ 
using (possibly entangled) measurements and make a decision based on the outcomes.
This ``store-and-measure'' strategy is vividly illustrated 
in Fig.~\ref{fig:multicopy-verification-framework} for $k=2$.
The test will be repeated $M$ times and the total number of consumed states is $Mk$.
We designate this quantum memory-assisted strategy as an \emph{$(n, k, d)$ verification strategy}.
Its crucial distinction from standard verification strategies,
comprehensively reviewed in~\cite{Yu_2022}, 
lies in the latter's absence of quantum memory assistance. 
In our notation, these standard strategies fall under the category of $(n,1,d)$ strategies.
In the good case, the overall state stored in the quantum memories
admits a tensor product structure: $\ket{\Psi}:=\bigotimes_{r = 1}^k \ket{\psi}^{(r)}$, 
where the superscript $r$ represents the $r$-th copy in the quantum memory. 
The verifiers perform a local binary measurement $\{T_\ell,\1-T_\ell\}$ 
such that state $\ket{\Psi}$ passes the test with certainty. 
In the bad case, we assume that the $k$ states produced by the quantum device
are independent, indicating that the fake state 
in the composite space $\cH^{nk}$ has the form
\begin{align} \label{eq:Hugefakestate}
    \xi =  \bigotimes_{r = 1}^{k} \sigma^{(r)},
\end{align}
where each $\sigma^{(r)}$ satisfies $\bra{\psi}\sigma^{(r)}\ket{\psi}\leq 1 - \varepsilon$. 
Correspondingly, the maximal probability that the fake state $\xi$ 
in the bad case can pass the test is
\begin{align}\label{eq:pm}
    p(\Omega) 
:= \max_{\bra{\psi} \sigma^{(r)} \ket{\psi} \leq 1 -\varepsilon}
    \tr\left[\Omega\left (\bigotimes_{r = 1}^{k} \sigma^{(r)}\right)\right].    
\end{align}
The minimum required number of measurements to saturate the worst-case failure probability, 
denoted as $M_m(\Omega)$, is given by $M_m(\Omega) = \ln \delta / \ln p(\Omega)$. 
Thus, the total number of copies consumed by the verification strategy $\Omega$ satisfies
\begin{align}
   N_m(\Omega) = kM_m(\Omega) = \frac{k \ln \delta}{\ln p(\Omega)}. \label{eq:Nm}
\end{align}
The verifiers' objective is to design efficient memory-assisted strategies $\Omega$ 
that minimize the number of copies consumed.

\begin{figure}[t]  
    \centering
    \includegraphics[width=0.96\linewidth,keepaspectratio]{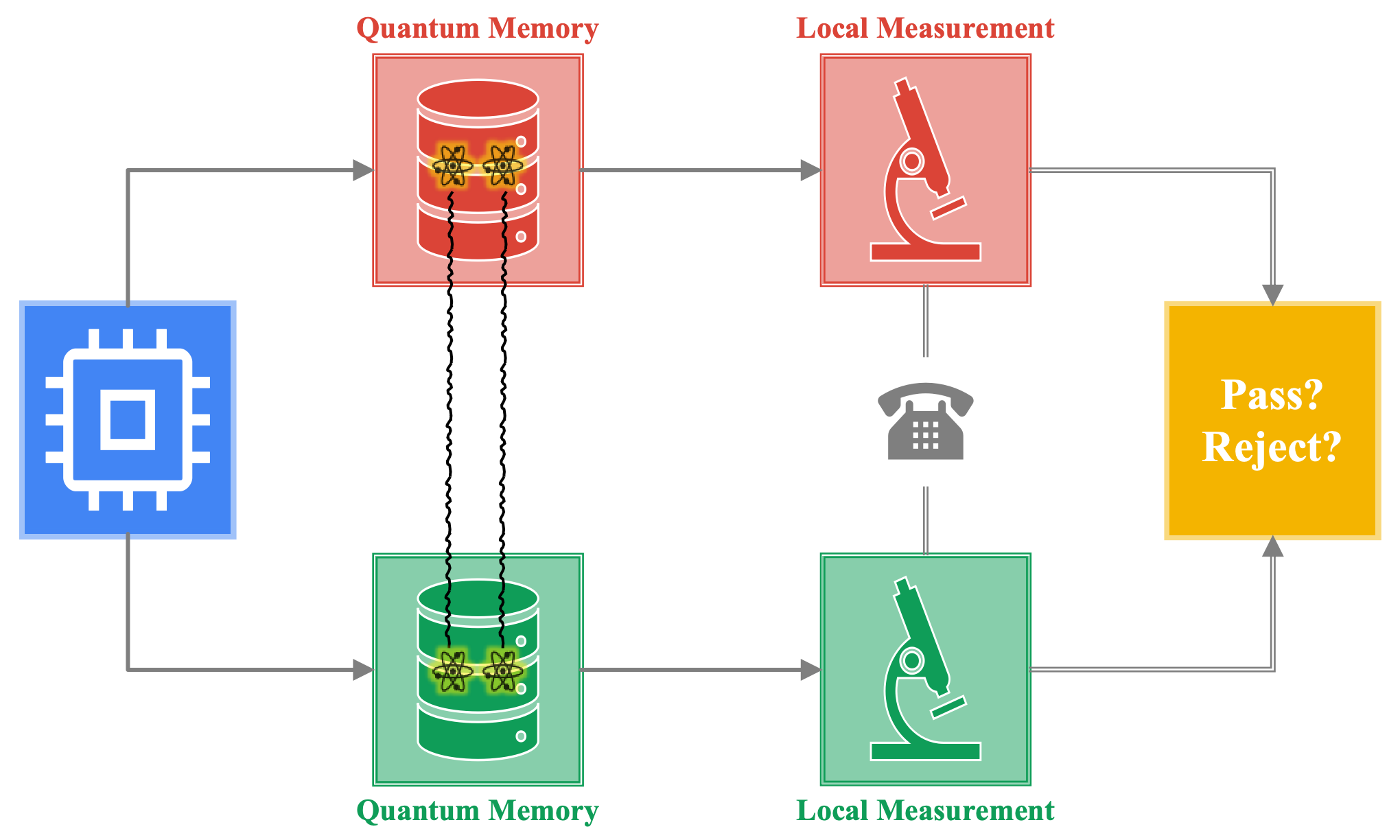}
    \caption{Schematic view of quantum memory assisted state verification.
            In this $(2,2,d)$ strategy, 
            the verifiers store two copies of quantum states (represented by atoms) 
            in their local quantum memories. 
            They then agree on local measurements via classical communication 
            and perform these measurements on their respective qudits. 
            Finally, they make a ``pass/reject'' decision from the measurement outcomes.}
    \label{fig:multicopy-verification-framework}
\end{figure}

\prlsection{Two-copy verification strategy.}We analytically solve the maximization problem 
in Eq.~\eqref{eq:pm} for the case of $k=2$, 
yielding an exact analytic formula for optimizing two-copy state verification.
First of all, we simplify the form of the optimisation in Eq.~\eqref{eq:pm}.
Regarding the permutation invariant nature of the verifiers, we show that 
it is best to consider verification strategies that are symmetric with respect to 
the two state copies; i.e., $\bF_{1\leftrightarrow2}\Omega\bF_{1\leftrightarrow2}=\Omega$,
where $\bF_{1\leftrightarrow2}$ is the swap operator between the first and second copy.
Regarding Eq.~\eqref{eq:pm}, we make the following useful observations:
(a) It suffices to optimize over pure fake states; and
(b) If the quantum device is not too bad, i.e., 
there exists an \emph{insurance infidelity} $\varepsilon_{\max}\left(\Omega\right)\geq\varepsilon$ ,depending on verification strategy,
such that $\bra{\psi}\sigma\ket{\psi}\geq 1-\varepsilon_{\max}\left(\Omega\right)$ for all $\sigma$, 
it is then suffices to consider fake states $\sigma$ for which $\bra{\psi}\sigma\ket{\psi}=1-\varepsilon$.
We prove these observations in Appendix~\ref{app:two-cp_optimization}, 
where we elaborate the significance and bounds of 
the insurance infidelity parameter $\varepsilon_{\max}\left(\Omega\right)$.
Note that for $\Omega$ satisfying certain conditions, such as the strategy for graph states described below, $\varepsilon_{\max}(\Omega) \gg \varepsilon$, thus this insurance can be tested almost at no cost.
We introduce the following two projectors:
\begin{align}\label{eq:projectors}
    \bP_s := \frac{\bF_{1\leftrightarrow2} + \mathbb{I}_{12}}{2},\quad
    \bP_\psi := \proj{\psi} \otimes (\mathbb{I} - \proj{\psi}),
\end{align}
which are useful in deriving the analytic formula. 
Note that $\mathbb{P}_s$ is the projector onto the symmetric subspace of $\cH^n\otimes\cH^n$.
For any symmetric two-copy verification strategy $\Omega$, 
define the doubly projected operator $\Omega_\star := 2\bP_\psi\bP_s\Omega \bP_s\bP_\psi$.
Let $\lambda_{\star}(\Omega)$ be the maximal eigenvalue of the projected operator $\Omega_\star$. 
We show that, $\lambda_{\star}$ is the intrinsic property of $\Omega$ which underpins 
$\Omega$'s verification efficiency, as elucidated in the ensuing theorem.
The proof can be found in Appendix~\ref{app:two-cp_optimization}. 

\begin{theorem}
\label{theorem:optimization_target}
When $\lambda_{\star}(\Omega)<1$ and 
the existence of insurance fidelity $\varepsilon_{\max}\left(\Omega\right)$ is guaranteed, it holds that
\begin{align}\label{eq:two-copy-probability}
p(\Omega) = 1 -2(1 - \lambda_{\star}(\Omega))\varepsilon + \cO(\varepsilon^{1.5}).
\end{align}
Correspondingly, the sample complexity of $\Omega$ is given by
\begin{align}\label{eq:tensorNm}
   N_m(\Omega) 
= \frac{2\ln \delta}{\ln p(\Omega)} 
\approx \frac{1}{(1-\lambda_{\star}(\Omega))\varepsilon } \ln \frac{1}{\delta}.
\end{align}
\end{theorem}

Comparing Eqs.~\eqref{eq:onecpN} and~\eqref{eq:tensorNm}, 
we see that it is $\lambda_{\star}(\Omega)$, instead of $\lambda_2(\Omega)$, 
that determines the sample complexity of $\Omega$ in the memory assisted scenario.
For the tensor product of single-copy globally optimal strategies, 
$\Omega_g = \proj{\psi}^{\ox 2}$, we find that $\lambda_{\star}(\Omega_g)=0$, 
implying a sample complexity of $1/\varepsilon\ln1/\delta$.
This confirms that, in this specific case, quantum memory assistance cannot surpass 
the ultimate bound established by entangled measurements. 
Similarly, for a tensor product strategy $\Omega = \Omega_l\otimes\Omega_l$, 
where $\Omega_l$ is any single-copy local verification strategy and quantum memories are absent,
$\lambda_{\star}(\Omega)=\lambda_2(\Omega_l)$, reducing precisely to the single-copy case. 
These examples demonstrate the alignment of our findings with existing results.
Extending Theorem~\ref{theorem:optimization_target} for arbitrary $k$ 
is possible through generalized versions of $\mathbb{P}_s$ and $\mathbb{P}_{\psi}$.
However, two-copy verification strategies already showcase the potential to achieve 
globally optimal efficiency as we will show in the following examples.
Moreover, the fidelity and coherence time requirements of quantum memory devices 
become increasingly stringent with larger $k$, 
potentially hindering their feasibility for practical applications beyond a certain threshold.

\prlsubsection{Graph states.}
As paradigmatic examples of quantum states which exhibit genuine multipartite entanglement, 
graph states are hold central importance in quantum computation and 
information due to their unique entanglement structure~\cite{raussendorf2001one,guhne2005bell,Broadbent2009Universal,perseguers2013distribution}.
A graph state is associated with a graph $G = (V,E)$. 
It can be prepared through Hadamard gates on qubit vertices in $V$ followed by control-Z gates on edges in $E$.
A simple example of graph state is $G_0$: 
\tikz \graph [nodes={circle,fill,inner sep=0.01pt,scale=0.6}, edges={-}] {1 -- 2 -- 3};, 
whose corresponding graph state is:
\begin{align}
    \ket{G_0} &= \frac{1}{\sqrt{8}} (\ket{000} + \ket{100} +\ket{010} \notag \\
              &\qquad - \ket{110} + \ket{001} + \ket{101} - \ket{011} + \ket{111}).
\end{align}

We leverage Theorem~\ref{theorem:optimization_target} to construct 
a two-copy verification strategy for arbitrary multi-qubit graph state $\ket{G}$,
demonstrating that even moderate quantum memory usage can boost the QSV efficiency to global optimality.

\begin{figure}[!hbtp]
    \centering
    \includegraphics[width=1.0\linewidth,keepaspectratio]{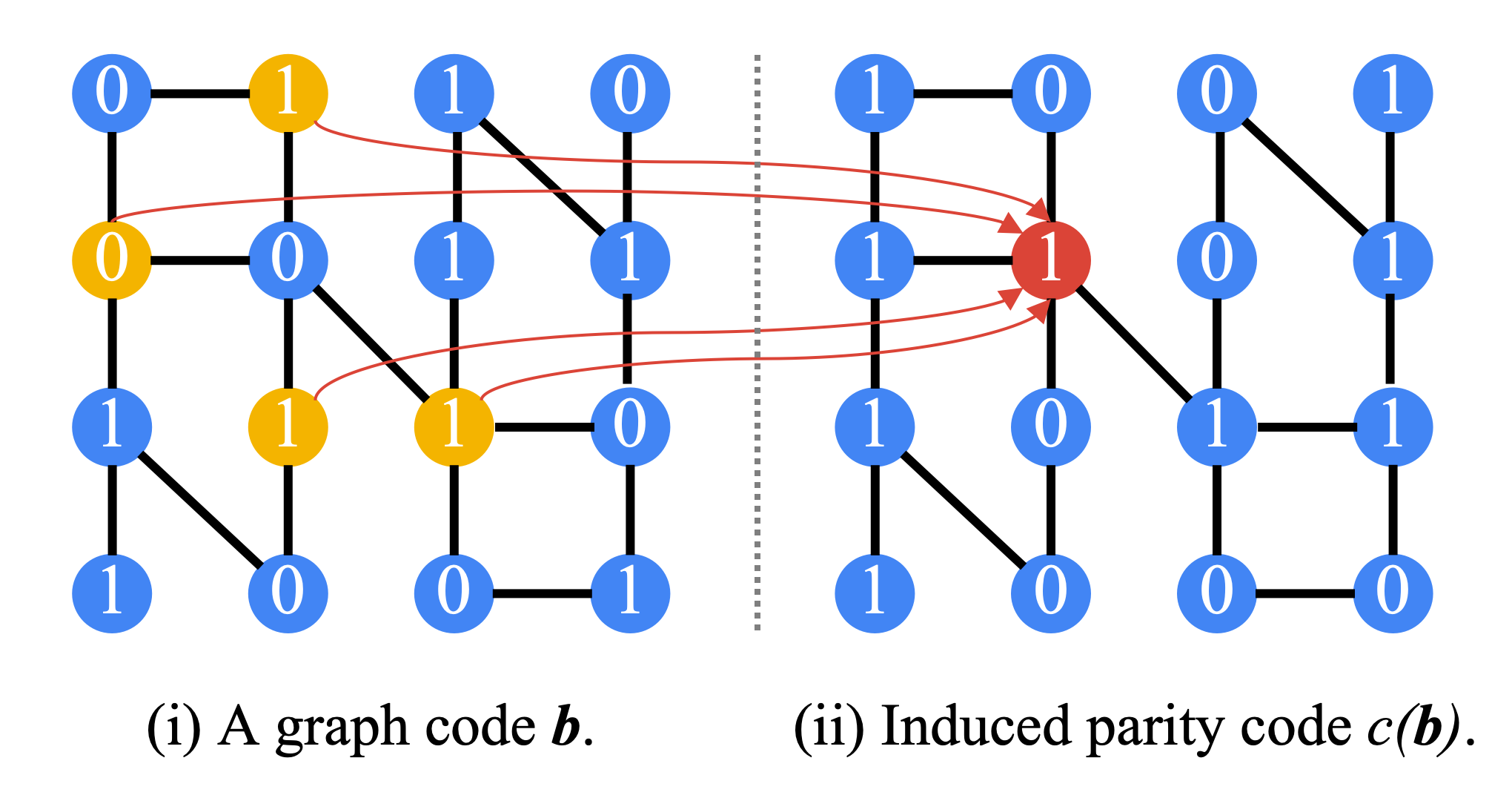}
   \caption{Schematic view of a graph code $\bm{b}$ of a graph 
            and its induced parity code $c(\bm{b})$. 
            The binary value of a vertex (red vertex) in the induced parity code 
            is given by the summation modulus $2$ of the values of its adjacent 
            vertices (yellow vertices) in the graph code $\bm{b}$.}
    \label{fig:codeparity}
\end{figure}

To formally describe our two-copy verification strategy for graph states, 
we begin by introducing the concept of graph code of a graph $G=(V,E)$.
Let $n=\vert V\vert$ be the number of vertices.
A \emph{graph code} $\bm{b}\in\{0,1\}^n$ is an $n$-bit binary string 
that assigns the binary value $\bm{b}_v\in \{0,1\}$ to vertex $v\in V$. 
Fig.~\ref{fig:codeparity}(i) visualizes a graph code of $G$ for example. 
Each graph code $\bm{b}$ uniquely induces a \textit{parity code} $c(\bm{b})\in\{0,1\}^n$,
where the binary string map $c:\{0,1\}^n\to\{0,1\}^n$ is defined as
$c_u(\bm{b}) := \sum_{v\in V, u\sim v} \bm{b}_v\;(\mathrm{mod}\;2)$,
$c_u$ is the value of vertex $u$, and $u\sim v$ means that $u$ is adjacent to $v$. 
An illustrative example is presented in Fig.~\ref{fig:codeparity}(ii). 
Let $\ket{\Phi_{00}}:=(\ket{00}+\ket{11})/\sqrt{2}$ be the standard two-qubit Bell state.
A binary code pair $(z, x)$ induces a locally transformed Bell state via
\begin{align}\label{eq:transformed-Bell-state}
    \ket{\Phi_{zx}} := (\mathbb{I}\otimes X^{x}Z^{z})\ket{\Phi_{00}},
\end{align}
where $X$ and $Z$ are the Pauli operators.
Our two-copy strategy for $\ket{G}$ involves only 
one binary measurement $\{\Omega_g,\mathbb{I}-\Omega_g\}$,
where $\Omega_g$ corresponding to passing the test is defined as
\begin{align}
\label{eq:left_graph_strategy}
\Omega_g &= \sum_{\bm{b}\in\{0,1\}^n} 
            \bigotimes_{j=1}^{n} 
            \proj{\Phi_{c_j(\bm{b})\bm{b}_j}}_{O_jO_j'},
\end{align}
where $O_j,O_j'$ represent two qubits held by the $j$-th verifier.
The verification strategy carries out as follows. 
In each test, the verifiers first store two copies of the states.
Then, the $j$-th verifier measures his qubits $O_jO_j'$ with the 
Bell measurement $\{\proj{\Phi_{zx}}\}_{x,z\in\{0,1\}}$ 
and records the outcome as $\bm{b}_j = x$ and $\bm{b}'_j = z$. 
Finally, they classically communicate the outcomes  
and obtain two graph codes $\bm{b},\bm{b}'$ of the graph $G$.
The states pass the test if and only if $\bm{b} = c(\bm{b}')$.

 \begin{figure}[!hbtp]
    \centering
    \includegraphics[width=0.96\linewidth]{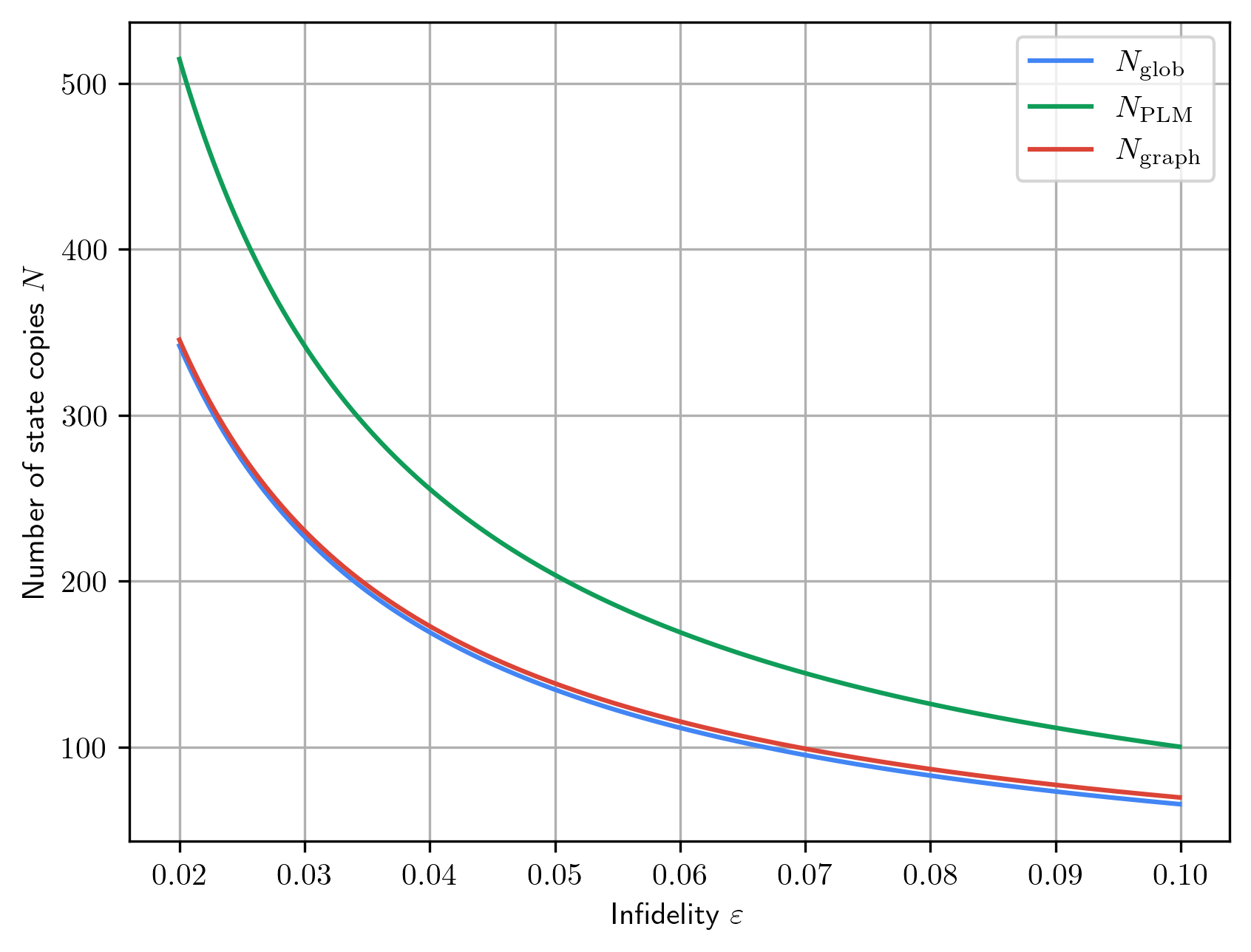}
    \caption{Comparison of the total number of state copies required to verify 
        the Bell state for different strategies as a function of the infidelity $\varepsilon$,
        where $\delta=0.001$. Here, 
        $N_{\rm graph}$ is the sample complexity of 
        our proposed two-copy graph verification strategy,
        $N_{\rm PLM}$ is the sample complexity of the optimal strategy 
        by Pallister \emph{et al.}~\cite{pallister2018optimal},
        and $N_{\rm glob}$ is the sample complexity of the globally optimal strategy.} 
    \label{fig:comparestrategy}
\end{figure}
Regarding the performance of our two-copy verification strategy $\Omega_g$,
we prove in Appendix~\ref{app:graph_disentangled} that $\lambda_\star(\Omega_g)=0$ and 
thus its optimal efficiency is achieved with a sample complexity of 
$N_{\rm graph}(\Omega_g) \approx 1/\varepsilon \ln1/\delta$ using Eq.~\eqref{eq:tensorNm},
indicating that $\Omega_g$ achieves globally optimal efficiency.
To showcase its significant advantage, 
we compare its efficiency with the optimal single-copy verification strategy 
by Pallister \emph{et al.}~\cite{pallister2018optimal} 
on verifying the canonical Bell state $\ket{\Phi_{00}}$.
As shown in Fig.~\ref{fig:comparestrategy}, 
our two-copy strategy rapidly converges towards 
the globally optimal solution in the small $\varepsilon$ regime, 
reducing the sample complexity by $50\%$ compared to 
the optimal single-copy verification strategy. 
This demonstrates a remarkable improvement in 
verification efficiency assisted by quantum memory.
Note that our two-copy verification strategy for the Bell state 
bears similarities with the celebrated 
entanglement-swapping protocol~\cite{zukowski1993event,Halder2007Entangling},
an important component of quantum networks.

Several remarks are in order. 
First, the construction of the above two-copy verification strategy for graph states, 
whose details can be found in Appendices~\ref{app:channel}, \ref{app:proofeqGraph}, and \ref{app:graph_disentangled}, is conceptually insightful and potentially extensible. 
Briefly, we begin by establishing a equivalence between information-preserving channels 
and optimal strategies, converting the verification problem to a state discrimination problem. Subsequently, we demonstrate that graph states can be leveraged to 
locally implement control-$Z$ gates, capitalizing on their inherent entanglement structure. 
This allows us to construct a quantum channel which induces the aforementioned strategy.
Second, 
the consistent Bell measurement across different verifiers, a key feature of our two-copy strategy, offers significant advantages for conducting state verification in neutral atom-based quantum systems~\cite{Bluvstein2023}. 
This consistency simplifies the verification process as a global laser can be employed, 
leveraging the Rydberg blockade radius, to parallelly
execute Bell measurements on all qubit pairs without single addressing~\cite{Wang_singleadressing}.
Third, we illustrate in Appendix~\ref{appx:fidelity-estimation} that, 
the verification strategy can be adapted to accomplish fidelity estimation.
Let $\sigma$ and $\sigma'$ be the unknown states produced in two device calls.
If the target quantum device is guaranteed to produce independent states, 
it holds that
\begin{align}
p_s = \tr[\Omega_g(\sigma \otimes \sigma')] 
    = \bra{G}\sigma\ket{G}\bra{G}\sigma'\ket{G} + \cO(\varepsilon^2).
\end{align}
Thus, when $\varepsilon$ is sufficiently small, 
the average fidelity $\cF$ of the states $\sigma$ with the target state $\ket{G}$
can be estimated from the statistical average of the passing frequency $p_s$ 
via $\cF=\sqrt{p_{s}}$.

\prlsection{Dimension expansion.}In the two-copy verification, 
we analytically solved the maximization problem in Eq.~\eqref{eq:pm},
relating the verification efficiency to an intrinsic property of $\Omega$.
However, it is demanding to generalize the result to larger $k$. 
Inspired by the observation that every $k$-tensor state $\ket{\Psi}$ 
can be equivalently viewed as a single $n$-partite state with local dimension $d^k$, we present the dimension expansion method that construct $(n,k,d)$-QSV protocol 
 according to existing $(n,1,d^k)$-QSV protocol with unchanged effeciency. 
This ``dimension expansion'' from $d$ to $d^k$ leverages quantum memory 
and establish an equivalence between an $(n,1,d^k)$ verification strategy and an $(n,k,d)$ strategy.
Concretely, we relax the maximization problem in Eq.~\eqref{eq:pm}
by considering any quantum state $\xi$ in $\cH^{nk}$ 
satisfying the fidelity constraint $\bra{\Psi}\xi\ket{\Psi}\leq(1-\varepsilon)^k$, 
providing an upper bound for the worst-case passing probability $p(\Omega)$:
\begin{align}
    p(\Omega) 
\leq \max_{\bra{\Psi} \xi \ket{\Psi} \leq  (1- \varepsilon)^k} \tr[\Omega \xi]
= 1 -  (1-\lambda_2(\Omega))\varepsilon',\label{eq:Nb_pre}
\end{align}
where $\varepsilon':=1 - (1 - \varepsilon)^k$ 
and the equality follows from Eq.~\eqref{eq:pallister2018optimal}.
Because $\ln \delta < 0$, according to Eq.~\eqref{eq:Nm}, we obtain an upper bound on $N_m(\Omega)$:
\begin{align}\label{eq:Nb_pre2}
N_m(\Omega) \leq  \frac{1}{ (1-\lambda_2(\Omega)) \varepsilon} \ln \frac{1}{\delta} =: N_{{\rm de},k}(\Omega).
\end{align}
Interestingly, $N_{{\rm de},k}(\Omega)$ is completely 
determined by $\lambda_2(\Omega)$, analogous to the single-copy state verification case. 

When investigating quantum memory assisted state verification, 
we have imposed two critical properties: 
(i) Locality: the fake states generated by the quantum device are independent; and 
(ii) Trust: the quantum memories are faithful without experimental error.
If either property is violated, the $k$-copy fake state 
might possess quantum correlation. 
In this correlated case, a weaker verification task determines whether \(\tr(\xi \ket{\Psi}\!\bra{\Psi}) < (1 - \varepsilon)^k\) or \(\xi = \ket{\Psi}\!\bra{\Psi}\), as discussed in Appendix~\ref{appx:comparision}.
The constraint then relaxes to $\bra{\Psi}\xi\ket{\Psi}\leq(1-\varepsilon)^k$,
leading to $N_m=N_{{\rm de},k}$ as evident from 
Eqs.~\eqref{eq:Nb_pre} and~\eqref{eq:Nb_pre2}. 
This signifies $N_{{\rm de},k}$ as a fundamental upper bound on 
the efficiency of quantum memory assisted state verification.

\prlsubsection{GHZ-like states.}
We demonstrate the power of the dimension expansion technique in 
constructing verification strategies for a broad class of GHZ-like states, 
encompassing arbitrary bipartite qudit states and GHZ states as special cases.
Mathematically, a multi-qudit GHZ-like state is defined as
\begin{align}\label{eq:GHZ-like}
    \ket{\psi_{\rm GHZ}} :=\sum_{j = 0}^{d - 1}s_j\ket{j_1} \otimes \cdots \otimes \ket{j_n},
\end{align}
where $\{\ket{j_r}\}_j$ is an orthonormal basis of the $r$-th qudit, and
the non-negative coefficients $s_j$ are decreasingly sorted and satisfy $\sum_js_j^2=1$. 
Whenever $s_0<1$, the GHZ state is entangled.
Note that the $k$-th tensor of a GHZ-like state is still a GHZ-like state, but with different coefficients.

\begin{figure}[!hbtp]
    \centering
    \includegraphics[width=0.96\linewidth]{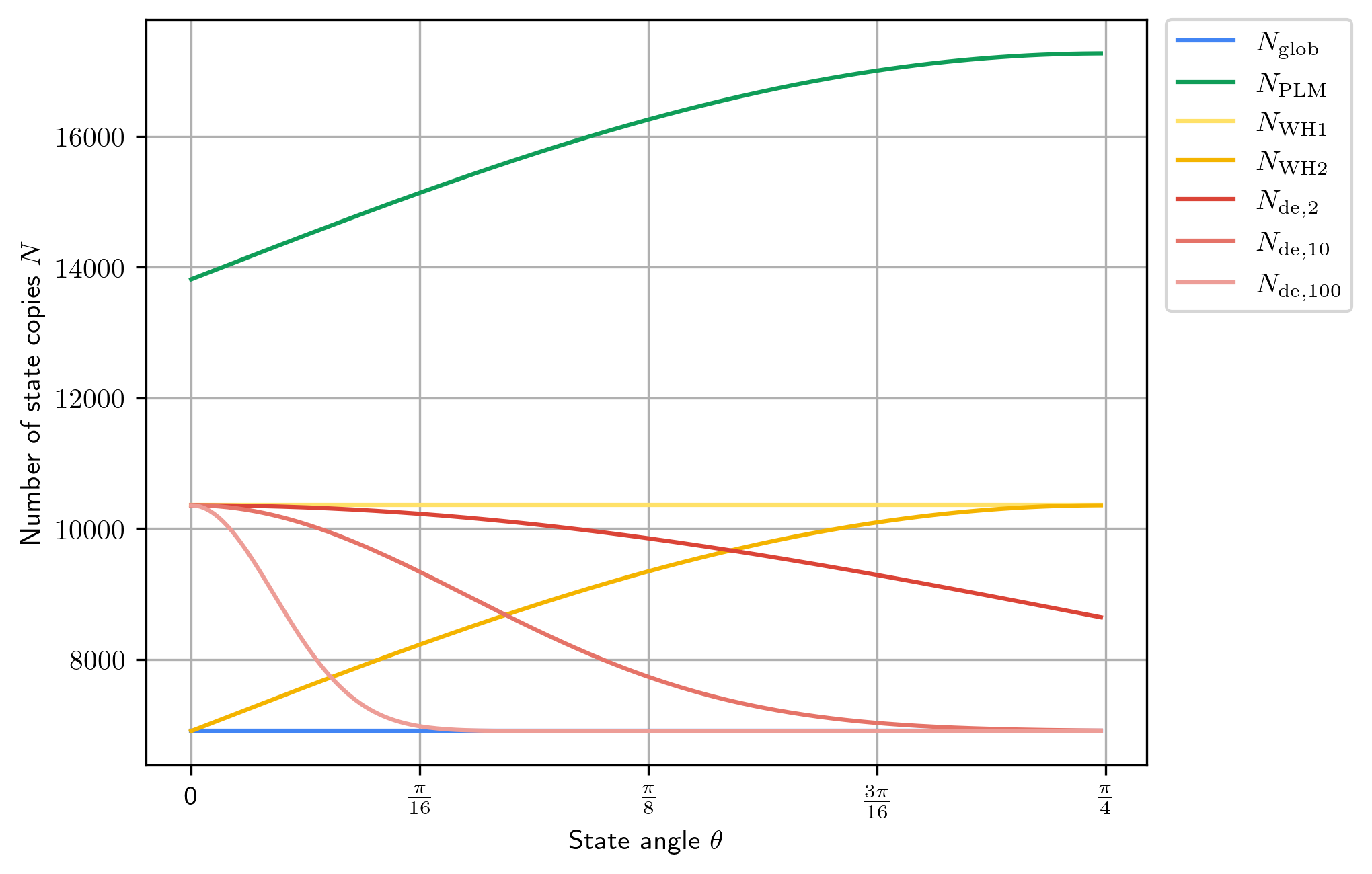}
    \caption{Comparison of the total number of state copies required to 
        verify the bipartite pure state $\ket{\psi} = \cos\theta \ket{00}+\sin\theta \ket{11}$
        for different strategies, where $\varepsilon=\delta=0.001$.
        Here, $N_{{\rm de},k}$ is the sample complexity of 
        our proposed dimension expansion strategy,
        $N_{\rm PLM}$ is the sample complexity of the optimal local strategy 
        by Pallister \emph{et al.}~\cite{pallister2018optimal},
        $N_{\rm WH1}$ and $N_{\rm WH2}$ are the sample complexities
        of the optimal one-way and two-way LOCC strategies by Wang and Hayashi~\cite{Wang_2019}, 
        and $N_{\rm glob}$ is the sample complexity of the globally optimal strategy.}
    \label{fig:arbver}
\end{figure}

Li \textit{et al.}~\cite{Li_2020} designed an efficient $(n,1,d)$ 
verification strategy $\Omega_{\rm LHZ}$ for GHZ-like states 
satisfying $\lambda_2(\Omega_{\rm LHZ}) = ((n - 1)s_0^2 + s_1^2)/(n+ (n - 1)s_0^2 + s_1^2)$.
The $(n,k,d)$-dimension expansion strategy for $\ket{\psi_{\rm GHZ}}$, 
which is deduced from the $(n,1,d^k)$ strategy 
for the $k$-th tensor product state $\ket{\psi_{\rm GHZ}}^{\otimes k}$, 
has the sample complexity
\begin{align}
    N_{{\rm de},k}(\ket{\psi_{\rm GHZ}})
&=  \frac{n +(n-1)s_0^{2k} + s_0^{2k-2}s_1^2}{n\varepsilon}
    \ln\frac{1}{\delta}.
\end{align}
One can verify that $N_{{\rm de},k}$ is monotonically decreasing in $k$; 
i.e., $k\geq k'$ implies $N_{{\rm de},k}(\ket{\psi_{\rm GHZ}})
\leq N_{{\rm de},k'}(\ket{\psi_{\rm GHZ}})$. 
Whenever $s_0<1$, indicating that the state is entangled, the dimension expansion 
strategy consistently outperforms the 
standard strategy with a \emph{net} benefit ratio of $s_0^{2k-2}$
and approaches the globally optimal efficiency when $k$ is sufficiently large.
Practically, the integer $k$ is upper bounded by $N_{{\rm de},k}$.

In Figure~\ref{fig:arbver}, the sample complexity required to verify 
the two-qubit state $\ket{\psi_\theta}=\cos\theta\ket{00}+\sin\theta\ket{11}$, 
being a special case of the GHZ-like states, is shown for different verification strategies.
We give the explicit construction of its verification strategy in Appendix~\ref{appx:dimension-expansion}.
The dimension expansion strategy derived here gives a remarkable 
improvement over the previously optimal local strategy
by Pallister \emph{et al.}~\cite{pallister2018optimal}
and optimal one-way LOCC strategy by Wang and Hayashi~\cite{Wang_2019} 
for the full range of $\theta\in(0,\pi/4)$, for the given values $\varepsilon$ and $\delta$.
Furthermore, it is evident from the figure that the dimension expansion strategy
becomes more and more advantageous as $k$ increases, eventually exceeding 
the optimal two-way LOCC strategy~\cite{Wang_2019} and
approaching the globally optimal efficiency, 
revealing the power of dimension expansion strategy.

\prlsection{Conclusions.}We have proposed a theoretical framework to 
quantitatively analyze the performance boost offered by quantum memories in quantum state verification.
Our work demonstrates that memory-assisted verification strategies significantly 
outperform non-assisted ones, 
with a remarkable finding that even just two copies suffice to 
achieve the theoretical limit of verification efficiency.
This superiority lies in the extended storage capacity, 
enabling the verifier to perform powerful entangled measurements within the memory. 

Many questions remain open.
Specifically, the analytic formula for two-copy verification and the optimal two-copy strategy 
for graph states might be generalized to wider scenarios 
with larger amount of quantum memories and arbitrary quantum states. 
However, deriving such solutions will likely require innovative techniques 
due to increased computational demands and higher state dimensions.

\prlsection{Acknowledgements.}
Part of this work was done when K. W. was a researcher and S.-Y. C. was a research intern at Baidu Research.
This work was supported by 
the National Key Research and Development Program of China (Grant No.~2022YFF0712800),
the Innovation Program for Quantum Science and Technology (Grant Nos.~2021ZD0301500 and 2021ZD0302901), and
the National Natural Science Foundation of China (Grant No.~62102388).

%


\makeatletter
\newcommand{\appendixtitle}[1]{\gdef\@title{#1}}
\newcommand{\appendixauthor}[1]{\gdef\@author{#1}}
\newcommand{\appendixaffiliation}[1]{\gdef\@affiliation{#1}}
\newcommand{\appendixdate}[1]{\gdef\@date{#1}}
\makeatother

\makeatletter%
\newcommand{\appendixmaketitle}{%
\begin{center}%
\vspace{0.4in}%
{\Large \@title \par}%
\end{center}%
\par%
}%
\makeatother%

\setcounter{secnumdepth}{2}
\appendix
\widetext
\newpage

\appendixtitle{\bf 
Supplemental Material for\\``Quantum memory assisted entangled state verification with local measurements''}
\appendixmaketitle
\vspace{0.2in}

The contents of the supplementary material are structured as follows: 
In Appendix~\ref{app:two-cp_optimization}, we articulate two optimization targets within the framework of two-copy verification, specifically substantiating Theorem~\ref{theorem:optimization_target}. 
In Appendix~\ref{app:channel}, we establish connections between optimal verification protocols and optimal information-preserving channels, essential for the development of a two-copy graph state verification protocol. 
In Appendix~\ref{app:proofeqGraph}, we prove the graph state disentangled equation presented in Theorem~\ref{Lemma:graph_disentangled}, 
a crucial component in constructing state-disentangled channels and applicable to tasks 
such as distributed quantum computation and fault-tolerant quantum computation. 
In Appendix~\ref{app:graph_disentangled}, we discuss the details concerning the optimal verification strategy for graph states and 
show that this strategy could be used in fidelity estimation.
In Appendix~\ref{appx:dimension-expansion}, we give an explicit construction of verification strategies based on the dimension expansion technique.
In Appenix~\ref{appx:comparision}, we provide a detailed comparison of our work with existing quantum-memory based verification strategies, 
highlighting the esential differences among these works.

\section{Two-copy verification strategy optimization}\label{app:two-cp_optimization}

In this section, we simplified the optimization in Eq.~({\color{beamer@blendedblue}4})
of the main text (MT) with $k = 2$ and prove the main Theorem~\ref{theorem:optimization_target}.

\subsection{Reduce to fake pure states}

First of all, one can easily prove that it suffices to optimize over pure states.
Here below, we use $\sigma$ to represent a fake state. A single $\sigma$ denotes a fake state in density matrix form, and by $\ket{\sigma}$, we mean a pure fake state $\ket{\sigma} = \sqrt{1 - \varepsilon}\ket{\psi} + \sqrt{\varepsilon} \ket{\psi^\perp}$. 
\begin{lemma}\label{lemma:pure-states}
The maximal passing probability $p(\Omega)$, 
defined in Eq.~({\color{beamer@blendedblue}4}) of MT, 
can be achieved among pure states, i.e.,
\begin{align}
    p(\Omega) 
=  \max_{\substack{\ket{\sigma},\vert\sigma'\rangle  \\ 
            \vert\braket{\psi}{\sigma}\lvert^2\leq 1-\varepsilon \\
            \vert\bra{\psi}\sigma'\rangle\vert^2\leq 1-\varepsilon}}
            \tr[\Omega(\proj{\sigma}\otimes\proj{\sigma'})].
\end{align}
\end{lemma}
\begin{proof}
We noted that this proof will be correct even if the fake state is classical-correlated. Since the maximum condition only reach on the product states
without classical correlation.

Accoroding to Eq.~({\color{beamer@blendedblue}4}), we have
\begin{align}\label{eq:appx-tmp}
    p(\Omega) 
=  \max_{\substack{\sigma,\sigma'  \\ 
        \bra{\psi}\sigma\ket{\psi}\leq 1 - \varepsilon \\
        \bra{\psi}\sigma'\ket{\psi}\leq 1 - \varepsilon}}
        \tr[\Omega(\sigma\otimes\sigma')].
\end{align}
We prove by contradiction that Eq.~\eqref{eq:appx-tmp} can be optimized over pure states.
Assume that two mixed states $\sigma$ and $\sigma'$ achieve Eq.~\eqref{eq:appx-tmp};
i.e., $p(\Omega)=\tr[\Omega(\sigma\otimes\sigma')]$.
Notice that the set of fake states
$\cS := \{\sigma\mid\bra{\psi}\sigma\ket{\psi}\leq 1 - \varepsilon\}$ is a convex set. 
Subsequently, the set of pure states
$\cP := \{\ket{\sigma} \mid \vert\braket{\psi}{\sigma}\vert^2 \leq 1 - \varepsilon\}$ 
contain the extreme points of the set $\cS$.
Given that both $\sigma,\sigma'\in\cS$, 
it is always possible to identify two pure-state decompositions 
\begin{align}
    \sigma = \sum_j \alpha_j\proj{\sigma_j},\qquad  
    \sigma' = \sum_k \beta_j\proj{\sigma'_k},
\end{align}
such that $\sum_j\alpha_j=1$, $\sum_k\beta_k=1$, 
and $\ket{\sigma_j},\ket{\sigma'_k}\in\cP$ for all $j$ and $k$, 
i.e., they are the extreme points within the set $\cP$.
Let $j_\star$ and $k_\star$ be the two indices whose corresponding pure states 
$\ket{\sigma_{j_\star}}$ and $\ket{\sigma'_{k_\star}}$
achieve the following maximization:
\begin{align}
    \tr[\Omega(\proj{\sigma_{j_\star}}\otimes\proj{\sigma'_{k_\star}})]
= \max_{j,k}\tr[\Omega(\proj{\sigma_j}\otimes\proj{\sigma'_k})].
\end{align}
Then the passing probability was evaluated to
\begin{align}
  p(\Omega) 
= \tr[\Omega(\sigma\otimes\sigma')]
&= \sum_{jk}\alpha_j\beta_k\tr\left[\Omega(\proj{\sigma_j}\otimes\proj{\sigma'_k})\right] \\
&\leq \sum_{jk}\alpha_j\beta_k
      \tr[\Omega(\proj{\sigma_{j_\star}}\otimes\proj{\sigma'_{k_\star}})] \\
&= \tr[\Omega(\proj{\sigma_{j_\star}}\otimes\proj{\sigma'_{k_\star}})].
\end{align}
That is to say, we can always identify two pure states---$\vert\sigma_{j_\star}\rangle, \vert\sigma'_{k_\star}\rangle \in \mathcal{S}$---that lead to a passing probability 
larger than $\tr[\Omega(\sigma\otimes\sigma')]$, leading to a contradiction. We are done.
\end{proof}

\vspace{0.1in}
Thanks to Lemma~\ref{lemma:pure-states}, 
we can restrain the fake state to the tensor product form of pure states as below:
\begin{align}\label{eq:appx-fake-states}
  \ket{\sigma}\otimes\ket{\sigma'}
= \sqrt{(1-\varepsilon_r)(1-\varepsilon_r')}\ket{\psi\psi}
+ \sqrt{(1-\varepsilon_r)\varepsilon_r'}\ket{\psi\psi'^\perp}
+ \sqrt{\varepsilon_r(1-\varepsilon_r')}\ket{\psi^\perp\psi}
+ \sqrt{\varepsilon_r\varepsilon_r'}\ket{\psi^\perp\psi'^\perp},
\end{align}
where $\varepsilon_r,\varepsilon'_r\geq\varepsilon$ 
and $\ket{\psi^\perp}, \ket{\psi'^\perp}$ are pure states orthogonal to $\ket{\psi}$.
Correspondingly, the passing probability evaluates to
\begin{align}
\bra{\sigma\sigma'}\Omega\ket{\sigma\sigma'}=&\sqrt{(1-\varepsilon_r)(1-\varepsilon_r')(1-\varepsilon_r)(1-\varepsilon_r')}\bra{\psi\psi}\Omega\ket{\psi\psi} + \sqrt{(1-\varepsilon_r)(1-\varepsilon_r')\varepsilon_r(1-\varepsilon_r')}\bra{\psi\psi}\Omega\ket{\psi\psi'^\perp} \notag \\
&+ \sqrt{(1-\varepsilon_r)(1-\varepsilon_r')\varepsilon_r(1-\varepsilon_r')}\bra{\psi\psi}\Omega\ket{\psi^\perp\psi}
+ \sqrt{(1-\varepsilon_r)(1-\varepsilon_r')\varepsilon_r\varepsilon_r'}\bra{\psi\psi}\Omega\ket{\psi^\perp\psi'^\perp} \notag\\
&+ \sqrt{(1-\varepsilon_r)\varepsilon_r'(1-\varepsilon_r)(1-\varepsilon_r')}\bra{\psi\psi'^\perp}\Omega\ket{\psi\psi} 
+ \sqrt{(1-\varepsilon_r)\varepsilon_r'(1-\varepsilon_r)\varepsilon_r'}\bra{\psi\psi'^\perp}\Omega\ket{\psi\psi'^\perp} \notag \\
& +  \sqrt{(1-\varepsilon_r)\varepsilon_r'\varepsilon_r(1-\varepsilon_r')}\bra{\psi\psi'^\perp}\Omega\ket{\psi^\perp\psi} 
 + \sqrt{(1-\varepsilon_r)\varepsilon_r'\varepsilon_r\varepsilon_r'}\bra{\psi\psi'^\perp}\Omega\ket{\psi^\perp\psi'^\perp} \notag \\
&+ \sqrt{\varepsilon_r(1-\varepsilon_r')(1-\varepsilon_r)(1-\varepsilon_r')}\bra{\psi^\perp\psi}\Omega\ket{\psi\psi} 
+ \sqrt{\varepsilon_r(1-\varepsilon_r')(1-\varepsilon_r)\varepsilon_r'}\bra{\psi^\perp\psi}\Omega\ket{\psi\psi'^\perp} \notag \\
&  + \sqrt{\varepsilon_r(1-\varepsilon_r')\varepsilon_r(1-\varepsilon_r')}\bra{\psi^\perp\psi}\Omega\ket{\psi^\perp\psi} 
 + \sqrt{\varepsilon_r(1-\varepsilon_r')\varepsilon_r\varepsilon_r'}\bra{\psi^\perp\psi}\Omega\ket{\psi^\perp\psi'^\perp} \notag \\
&+ \sqrt{\varepsilon_r\varepsilon_r'(1-\varepsilon_r)(1-\varepsilon_r')}\bra{\psi^\perp\psi'^\perp}\Omega\ket{\psi\psi} 
+ \sqrt{\varepsilon_r\varepsilon_r'(1-\varepsilon_r)\varepsilon_r'}\bra{\psi^\perp\psi'^\perp}\Omega\ket{\psi\psi'^\perp} \notag \\
& + \sqrt{\varepsilon_r\varepsilon_r'\varepsilon_r(1-\varepsilon_r')}\bra{\psi^\perp\psi'^\perp}\Omega\ket{\psi^\perp\psi} 
 + \sqrt{\varepsilon_r\varepsilon_r'\varepsilon_r\varepsilon_r'}\bra{\psi^\perp\psi'^\perp}\Omega\ket{\psi^\perp\psi'^\perp} \\
 &=:p(\Omega,\varepsilon_r,\varepsilon_r',\psi,\psi').
\end{align}

Any reasonable two-copy verification strategy $\Omega$ must satisfy the following two conditions:
\begin{align}
    \Omega\ket{\psi} \otimes \ket{\psi}
&=  \ket{\psi} \otimes \ket{\psi} \label{eq:critical1},\\
\Omega 
&= \bF_{1\leftrightarrow2} \Omega \bF_{1\leftrightarrow2}. \label{eq:critical2} 
\end{align}
The first property is justifiable because, for any $\Omega$ failing to meet this condition, the inequality $N_m(\Omega) \geq 2p(1-p)1/\varepsilon^2\ln1/\delta$ is valid when $\varepsilon$ is sufficiently small~\cite{Pallisterthesis}. Here, $p = \tr[\Omega(\proj{\psi} \otimes \proj{\psi})]\neq1$. 
The quadratic nature of $\varepsilon^2$ leads to a considerably higher sampling complexity compared 
to those satisfying the first condition when $\varepsilon$ is small. 
The second condition is rationalized by the fact that the verifier can employ 
classical randomness to execute the LOCC 
strategy $\frac{1}{2}( \Omega + \bF_{1\leftrightarrow2} \Omega  \bF_{1\leftrightarrow2})$ based 
on any existing LOCC strategy $\Omega$ that might not fulfill the second condition. 

\subsection{Discussion on the insurance infidelity}

In this section, we exclusively discusses the existence condition and upper bound of the insurance infidelity $\varepsilon_{\max}\left(\Omega\right)$.
\begin{proposition}
\label{prop:exist_insurance}

Let $\ket{\psi}$ represent the target state and $\Omega$ denote its two-copy verification strategy, which exhibits symmetry under copy exchange. We define $\gamma_{\star}(\Omega)$ and $\xi_{\star}(\Omega)$ as the maximum eigenvalues of the operators $\mathbb{P}_{\psi} \mathbb{F}_{1\leftrightarrow2} \Omega \mathbb{P}_{\psi}$ and $\mathbb{P}_{\psi} (\mathbb{F}_{1\leftrightarrow2}/2 + \mathbb{I}_{12}) \Omega \mathbb{P}_{\psi}$, respectively, where $\mathbb{F}_{1\leftrightarrow2}$ and $\mathbb{P}_{\psi}$ are defined in Eq.~\eqref{eq:projectors} of MT. When $\varepsilon$ is sufficiently small ($\varepsilon \ll 1$) and it is guaranteed that $\xi_{\star}(\Omega) + \gamma_{\star}(\Omega)/2 < 1$, for any choice of $\ket{\psi^\perp}$ and $\ket{\psi'^\perp}$, the function:
\begin{align}
  p(\varepsilon_r',\varepsilon_r,\ket{\psi^\perp},\ket{\psi'^\perp}) =  \bra{\sigma\sigma'}\Omega\ket{\sigma\sigma'},
\end{align}
reaches its maximum at the point $(\varepsilon_r, \varepsilon'_r) = (\varepsilon, \varepsilon)$ within a local region $R = \{(\varepsilon_r, \varepsilon'_r) | \varepsilon_r, \varepsilon'_r > \varepsilon,\ \ \varepsilon_r + \varepsilon'_r < 2\varepsilon_{\max}\left(\Omega\right)\}$. Additionally, $\varepsilon_{\max}\left(\Omega\right)$, referred to as the insurance infidelity, is unrelated to $\ket{\psi^\perp}$ and $\ket{\psi'^\perp}$, and must satisfy either of the following conditions:
\begin{enumerate}
    \item If $\sqrt{\varepsilon} \ll \gamma_{\star}(\Omega)$, 
    then $\varepsilon_{\max}\left(\Omega\right) = 0.5\varepsilon + 0.5\varepsilon\left[\left(1 - \xi_{\star}(\Omega) + 0.5 \gamma_{\star}(\Omega)\right)/{\gamma_{\star}(\Omega) }\right]^2  >\varepsilon$.
    \item If $\sqrt{\varepsilon} \sim \gamma_{\star}(\Omega)$, then $\varepsilon_{\max}\left(\Omega\right) \gg \varepsilon$. 
\end{enumerate}
\end{proposition}

\begin{proof}
Given the sufficiently small nature of $\varepsilon$, we initially approximate $\varepsilon_r$ and $\varepsilon'_r$ as approximately equal to $\varepsilon$, resulting in the simplified expression for the passing probability:
\begin{align}
    p = 1 - \varepsilon_r - \varepsilon_r' + \bra{\psi \psi^\perp} \Omega \ket{\psi \psi^\perp}\varepsilon_r + \bra{\psi \psi'^\perp} \Omega \ket{\psi \psi'^\perp} \varepsilon_r'+ (\bra{\psi^\perp \psi} \Omega \ket{\psi \psi^\perp} + \bra{\psi'^\perp\psi } \Omega \ket{\psi \psi'^\perp}) \sqrt{\varepsilon_r \varepsilon_r'} + \cO(\varepsilon^{1.5}).
\end{align}
The leading orders dominate the behavior of the function $p$ in the vicinity of the $(\varepsilon, \varepsilon)$ region. Therefore, our task is to demonstrate that the leading term reaches a local maximum at the point $(\varepsilon, \varepsilon)$ under the constraint $\varepsilon_r, \varepsilon_r' > \varepsilon$. To facilitate this analysis, we introduce the variable transformation $(x, x') = (\sqrt{\varepsilon_r}, \sqrt{\varepsilon_r'})$, after which the leading term undergoes a transformation to:
\begin{align}
    p_{\rm lead} = 1 - (1-R)x^2 - (1- R') x'^2 + (B+B')xx',
    \label{eq:plead}
\end{align}
where $x, x' > \sqrt{\varepsilon}$ and
\begin{align}
    R &= \bra{\psi \psi^\perp} \Omega \ket{\psi \psi^\perp},~ R' = \bra{\psi \psi'^\perp} \Omega \ket{\psi \psi'^\perp}, \\
    B &=  \bra{\psi^\perp \psi} \Omega \ket{\psi \psi^\perp},~B'= \bra{\psi'^\perp\psi}\Omega \ket{\psi \psi'^\perp}.
\end{align}
We first noticed that:
\begin{align}
    &\frac{\partial p_{\rm lead}}{\partial x} = -2(1- R) x +(B+B')x', \\
    &\frac{\partial p_{\rm lead}}{\partial x'} = -2(1- R') x' +(B+B')x.
    \label{eq:plead_deri}
\end{align}
To achieve a local maximum at $(\sqrt{\varepsilon},\sqrt{\varepsilon})$ under the constraint $x, x' > \sqrt{\varepsilon}$, both derivatives at the point $x = x' = \sqrt{\varepsilon}$ must be less than zero for arbitrary $\ket{\psi^\perp}$ and $\ket{\psi'^\perp}$. This implies that:
\begin{align}
    \forall \ket{\psi^\perp},\ket{\psi'^\perp},\quad  1  > \frac{B}{2} + R + \frac{B'}{2},\quad 1  > \frac{B}{2} + R' + \frac{B'}{2}.
\end{align}
Subsequently, we establish two critical values for the operator $\Omega$ with respect to the quantum state $\ket{\psi}$
\begin{align}
    &\gamma_{\star}(\Omega) =  \max_{\ket{\psi^\perp}} \bra{\psi^\perp \psi} \Omega \ket{\psi \psi^\perp}, \\
    & \xi_{\star}(\Omega) =  \max_{\ket{\psi^\perp}} \left( \frac{1}{2}  \bra{\psi^\perp \psi} \Omega \ket{\psi \psi^\perp} + \bra{\psi \psi^\perp} \Omega \ket{\psi \psi^\perp} \right).
\end{align}
Utilizing these values, the local maximum condition is equivalent to the assertion that:
\begin{align}
    1 &> \max_{\ket{\psi^\perp}} \left( \frac{1}{2}  \bra{\psi^\perp \psi} \Omega \ket{\psi \psi^\perp} + \bra{\psi \psi^\perp} \Omega \ket{\psi \psi^\perp} \right) + \frac{1}{2} \max_{\ket{\psi'^\perp}} \left(\bra{\psi'^\perp \psi} \Omega \ket{\psi \psi'^\perp}\right) \\
    &=  \xi_{\star}(\Omega) + \frac{1}{2} \gamma_{\star}(\Omega).
\end{align}
In order to delineate the range of this local maximum, we initially assume that $\gamma_{\star}(\Omega) \gg \sqrt{\varepsilon}$. Subsequently, we designate the selections of $\psi^\perp$ and $\psi'^\perp$ and find the domain where the $p_{\rm lead}$ always decreases as both variables $x$ and $x'$ increased.
This region is delimited by two linear constraints:
\begin{align}
    \frac{\partial p_{\rm lead}}{\partial x} =-2(1- R) x +(B+B')x' < 0, \quad \frac{\partial p_{\rm lead}}{\partial x'} =-2(1- R) x' +(B+B')x < 0.
\end{align}
\begin{figure}[!hbtp]
    \centering
    \includegraphics[width=0.5\linewidth]{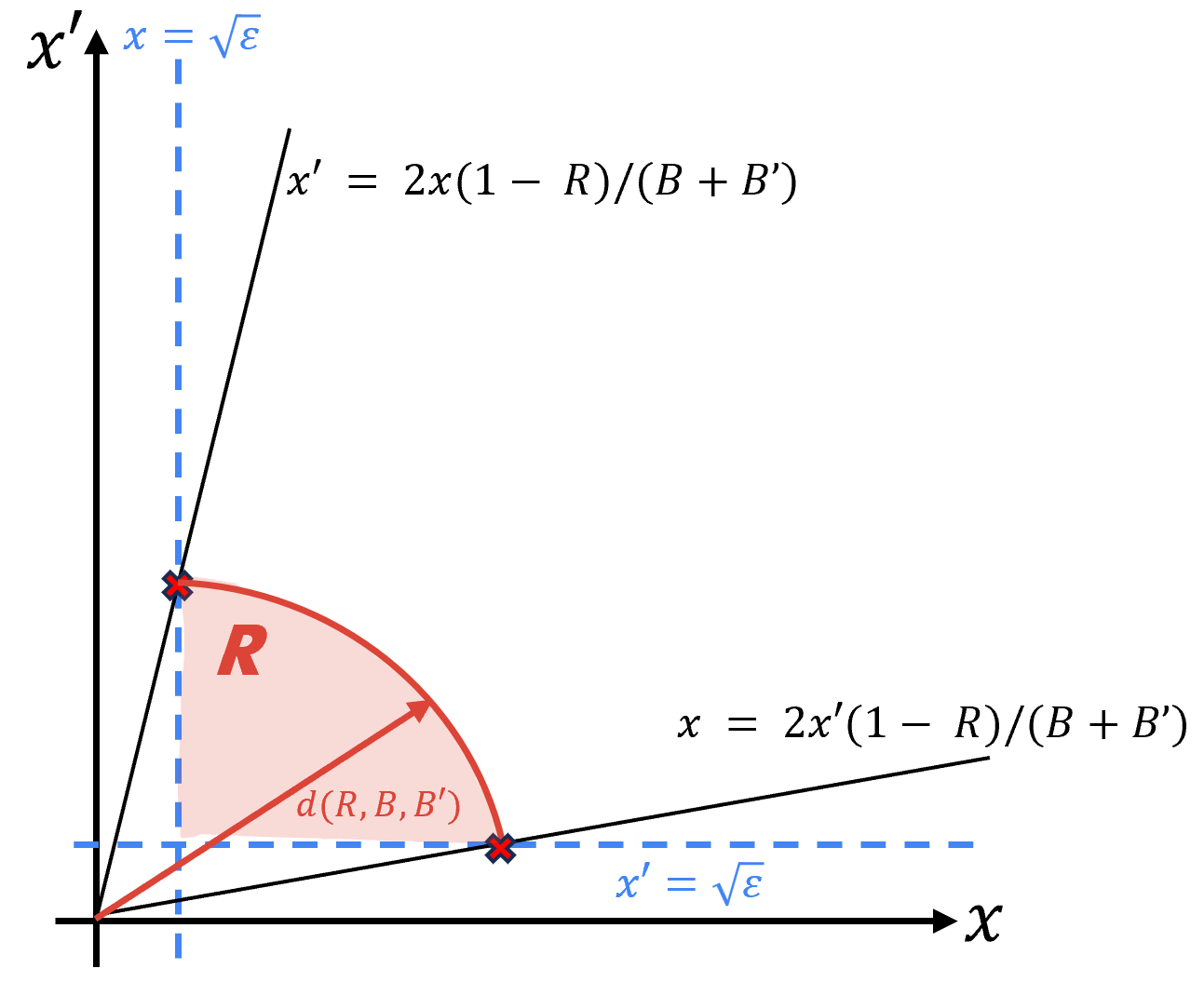}
    \caption{This figure show the region $R$ inside which $p_{lead}$ reach local maximum at point  $(\sqrt{\varepsilon},\sqrt{\varepsilon})$. The insurance infidelity the could be calculated from the intersection of line $x' = 2x(1- R)/(B+B')$ and $x = \sqrt{\varepsilon}$ .  }
    \label{fig:insurance_inf}
\end{figure}
The local maximum condition ensures that $2(1 - R) > B + B'$. Consequently, every point within the set $R(R,B,B') = \{(x,x')|x,x'>\sqrt{\varepsilon},\ \ x^2+x'^2 < d(R,B,B') \}$ should decrease as both $(x,x')$ increase, as depicted in Figure \ref{fig:insurance_inf}. Here, $d(R,B,B')$ is defined as follows:
\begin{align}
    d(R,B,B') &= \varepsilon + \varepsilon\left(1 + 2\frac{1 - (R + \frac{B}{2}) - \frac{B'}{2}}{B + B'}\right)^2 \\
    &> \varepsilon + \varepsilon\left(1 + 2\frac{1 - \max_{\ket{\psi^\perp}}(R + \frac{B}{2}) - \max_{\ket{\psi'^\perp}}\frac{B'}{2}}{\max_{\ket{\psi^\perp}}B + \max_{\ket{\psi'^\perp}}B'}\right)^2 \\
    &= \varepsilon + \varepsilon\left(1 + \frac{1 - \xi_{\star}(\Omega) - \frac{1}{2} \gamma_{\star}(\Omega)}{\gamma_{\star}(\Omega) }\right)^2 \\
    & =\varepsilon + \varepsilon\left(\frac{1 - \xi_{\star}(\Omega) + \frac{1}{2} \gamma_{\star}(\Omega)}{\gamma_{\star}(\Omega) }\right)^2.
    \label{eq:upperbound_of_d}
\end{align}
Hence, within the region $R(R,B,B')$, the function $p$ attains its maximum at the point $(x,x') = (\sqrt{\varepsilon}, \sqrt{\varepsilon})$. Given an arbitrary selection of $\psi^\perp$ and $\psi'^\perp$, their intersection is determined as follows:
\begin{align}
   R &= \bigcap_{\forall \psi^\perp, \psi'^\perp} R(R,B,B') \\ 
   &= \{(x,x')|x,x'>\sqrt{\varepsilon},\ \ x^2+x'^2 < \min_{\psi^\perp, \psi'^\perp}d(R,B,B') \}\\
   & = \{(\varepsilon_r, \varepsilon'_r) | \varepsilon_r, \varepsilon_r' > \varepsilon,\ \ \varepsilon_r + \varepsilon_r' < 2\varepsilon_{\max}\left(\Omega\right)\}.
\end{align}
Here, $\varepsilon_{\max}\left(\Omega\right) = \min_{\psi^\perp, \psi'^\perp} d(R,B,B')/2$. The upper bound of $d$ in Eq.~\eqref{eq:upperbound_of_d} provides the upper limit for $\varepsilon_{\max}\left(\Omega\right)$
\begin{align}
    \varepsilon_{\max}\left(\Omega\right) 
> \frac{1}{2} \varepsilon + \frac{1}{2}\varepsilon\left(\frac{1 - \xi_{\star}(\Omega) + \frac{1}{2} \gamma_{\star}(\Omega)}{\gamma_{\star}(\Omega) }\right)^2 > \varepsilon.
\end{align}
In the last inequality, we invoke the local maximum condition once more, expressed as $1 > \xi_{\star}(\Omega) + \frac{1}{2} \gamma_{\star}(\Omega)$.

For strategies that satisfy $\gamma_{\star}(\Omega) \sim \sqrt{\varepsilon} $, this upper bound is not valid. Other terms in the function $p$, such as $\bra{ \psi^\perp_m \psi_m^\perp} \Omega \ket{\psi \psi^\perp_m}$, must be considered. However, in this case, one can demonstrate that $| \varepsilon_{\max}\left(\Omega\right) - \varepsilon | \gg \varepsilon$ by recalculating the leading terms near $\varepsilon$:
\begin{align}
    p = 1 - \varepsilon_r - \varepsilon_r' + \bra{\psi \psi^\perp} \Omega \ket{\psi \psi^\perp}\varepsilon_r + \bra{\psi \psi'^\perp} \Omega \ket{\psi \psi'^\perp} \varepsilon_r' + \cO(\varepsilon^{1.5}).
\end{align}
Given the projective construction, we have $\bra{\psi\psi'^\perp}\Omega\ket{\psi\psi'^\perp}\leq 1$. Consequently, $(\varepsilon,\varepsilon)$ is the maximum in the region satisfying $|\varepsilon_r - \varepsilon| \sim \varepsilon$. This implies that $|\varepsilon_{\max}\left(\Omega\right) - \varepsilon| \gg \varepsilon$.

We can further simplify the expression of $\gamma_{\star}(\Omega)$:
\begin{align}
    \gamma_{\star}(\Omega) 
= \max_{\ket{\psi^\perp}}\bra{\psi \psi^\perp} \bF_{1\leftrightarrow2} \Omega  \ket{\psi \psi^\perp} 
=\max_{\ket{\Phi}}\bra{\Phi} \mathbb{P}_{\psi}\bF_{1\leftrightarrow2} \Omega  \mathbb{P}_{\psi} \ket{ \Phi},
\end{align}
where $\mathbb{P}_{\psi}=\proj{\psi} \otimes (\mathbb{I} - \proj{\psi})$.
Then, $\gamma_{\star}(\Omega)$ is the maximum eigenvalue of the operator $\mathbb{P}_{\psi} \bF_{1\leftrightarrow2} \Omega  \mathbb{P}_{\psi} $. Similarly,  $\xi_{\star}(\Omega)$ corresponds to the maximum eigenvalue of operator $\mathbb{P}_{\psi} (\bF_{1\leftrightarrow2}/2 + \mathbb{I}_{12}) \Omega  \mathbb{P}_{\psi} $. 

\end{proof}

\subsection{Proof of Theorem~\ref{theorem:optimization_target} in MT}

Now we prove Theorem~\ref{theorem:optimization_target} in MT.

\begin{theorem}[Refined version of Theorem~\ref{theorem:optimization_target} in the main text]
\label{theorem:optimization_target-app}
Let $\Omega$ be an arbitrary two-copy verification strategy which is symmetric under copy exchange, 
we define $\lambda_{\star}(\Omega)$ as the maximum eigenvalue 
of the operator $\Omega_\star := 2\bP_\psi\bP_s\Omega \bP_s\bP_\psi$,
where $\mathbb{P}_s$ and $\bP_\psi$ are defined in Eq.~\eqref{eq:projectors} of MT. 
When $\varepsilon$ is sufficiently small ($\varepsilon \ll 1$) and
the local maximum condition in Proposition~\ref{prop:exist_insurance} is satisfied with insurance infidelity $\varepsilon_{\max}\left(\Omega\right)$. Then
\begin{align}
  p(\Omega) 
= \max_{\substack{\ket{\psi^\perp},\ket{\psi'^\perp}\\ 
        \varepsilon_r,\varepsilon_r' \in [\varepsilon,\varepsilon_{\max}\left(\Omega\right)]}}      \bra{\sigma\sigma'}\Omega\ket{\sigma\sigma'} 
=  1 -2(1 - \lambda_{\star}(\Omega))\varepsilon + \cO(\varepsilon^{1.5}),
\end{align}
\end{theorem}

\begin{proof}
We establish an additional critical maximum value for the operator $\Omega$ and the quantum state $\ket{\psi}$:
\begin{align}
    &\lambda_{\star}(\Omega) = \max_{\ket{\psi^\perp}}\left(\bra{\psi \psi^\perp} \Omega \ket{\psi \psi^\perp} + \bra{\psi^\perp \psi} \Omega \ket{\psi \psi^\perp}\right). 
\end{align}
According to Proposition~\ref{prop:exist_insurance}, the existence of insurance infidelity guarantees that $\xi_{\star}(\Omega) + \frac{1}{2} \gamma_{\star}(\Omega) < 1$. Consequently, $\lambda_{\star}(\Omega) \leq \xi_{\star}(\Omega) + \frac{\gamma_{\star}(\Omega)}{2} < 1$.

Given the insurance infidelity $\varepsilon_{\max}\left(\Omega\right)$ and the set $R$ defined in Proposition~\ref{prop:exist_insurance}, we observe that the set $S = \{(\varepsilon_r, \varepsilon_r')|\varepsilon_r, \varepsilon_r' \in [\varepsilon, \varepsilon_{\max}\left(\Omega\right)]\}$ satisfies $S \subset R$. Therefore, $p(\Omega)$, being the maximum value within the region $S$ with respect to variables $\psi^\perp$, $\psi'^\perp$, $\varepsilon_r$, and $\varepsilon'_r$, is attained solely at the constraint $(\varepsilon_r, \varepsilon'_r)=(\varepsilon, \varepsilon)$.

Further optimization over $\psi$ and $\psi^\perp$ is as follows:
\begin{align}
    p_{\max} &= 1 - \max_{\ket{\psi^\perp}, \ket{\psi'^\perp}}[(1-R-B) + (1 - R' -B')]\varepsilon + \cO(\varepsilon^{1.5}) \\
    &= 1 -2(1 - \lambda_{\star}(\Omega))\varepsilon + \cO(\varepsilon^{1.5}).
\end{align}

Again, given that $[\Omega,\bF_{1 \leftrightarrow 2}] = 0$, we can further simplify the expression of $\lambda_{\star}(\Omega)$:
\begin{align}
    \lambda_{\star}(\Omega) &= \max_{\ket{\psi^\perp}}\bra{\psi \psi^\perp} (\bF_{1\leftrightarrow2}+\mathbb{I}_{12})\Omega \ket{\psi \psi^\perp} \\
    &=\max_{\ket{\Phi}}\bra{\Phi} \mathbb{P}_{\psi}(\bF_{1\leftrightarrow2}+\mathbb{I}_{12})\Omega \mathbb{P}_{\psi} \ket{ \Phi}\\
    &= \max_{\ket{\Phi}}\bra{\Phi} 2\mathbb{P}_{\psi}\mathbb{P}_s\Omega \mathbb{P}_s\mathbb{P}_{\psi}  \ket{ \Phi},
\end{align}
where
\begin{align}
\mathbb{P}_{\psi} = \proj{\psi} \otimes (\mathbb{I} - \proj{\psi}),\qquad
\mathbb{P}_s = \frac{1}{2}(\bF_{1\leftrightarrow2} + \mathbb{I}_{12}).
\end{align}
Then, $\lambda_{\star}(\Omega)$ is the maximum eigenvalue of the operator $\Omega_{\star} = 2\mathbb{P_{\psi}}\mathbb{P}_s\Omega \mathbb{P}_s\mathbb{P}_{\psi}  $. 
\end{proof}

\subsection{Demonstrative example: The simple tensor product case}

From Theorem~\ref{theorem:optimization_target-app}, 
we know that to verify an arbitrary target state $\ket{\psi}$, 
we need to achieve the following objectives:
(a) Construct families of local projective measurements that 
unconditionally accept $\ket{\psi} \otimes \ket{\psi}$ with certainty 
and exist ensurance infidelity $\varepsilon_{\max}\left(\Omega\right)$, where $\Omega$ is the corresponding strategy;
(b) Minimize $\lambda_{\star}(\Omega)$
while maintaining $\varepsilon_{\max}\left(\Omega\right)$ at a suitable value.

To benchmark the optimization tasks described above, 
we consider the strategy $\Omega = \Omega_l \otimes \Omega_l$, 
which is simply a tensor product of two single-copy strategies $\Omega_l$. 
The operator $\Omega_{\star}$ can be calculated as below:
\begin{align}
    \Omega_{\star} &= \frac{1}{2}\mathbb{P_{\psi}}(\bF_{1\leftrightarrow2} + \mathbb{I}_{12}) \Omega_l \otimes \Omega_l (\bF_{1\leftrightarrow2} + \mathbb{I}_{12})\mathbb{P}_{\psi} \\
    &= \mathbb{P_{\psi}} \Omega_l \otimes \Omega_l \mathbb{P_{\psi}} + \mathbb{P_{\psi}}\frac{\bF_{1\leftrightarrow2} \Omega_l \otimes \Omega_l + \Omega_l \otimes \Omega_l\bF_{1\leftrightarrow2}}{2} \mathbb{P_{\psi}}  \\
    & = \mathbb{P_{\psi}} \Omega_l \otimes \Omega_l \mathbb{P_{\psi}}  \\
    &= \proj{\psi} \otimes [(\mathbb{I} - \proj{\psi})\Omega_l (\mathbb{I} - \proj{\psi})],
\end{align}
where in the third equality we use the fact that
$\mathbb{P_{\psi}} \bF_{1\leftrightarrow2} \Omega_l \otimes \Omega_l \mathbb{P_{\psi}} = 0$.
Then $\lambda_{\star}(\Omega) = \lambda_2(\Omega_l)$. 
This reduces to the standard single-copy verification efficiency as expected. 
Calculations also show that $\gamma_{\star}(\Omega_l \otimes \Omega_l) = 0$, 
indicating that $\varepsilon_{\max}\left(\Omega\right) \gg \varepsilon$.
In the following appendix, we construct a non-trivial two-copy strategy $\Omega$ 
for graph states, which satisfies that 
$\lambda_{\star}(\Omega) = 0$, $\gamma_{\star}(\Omega)=0$, and $\varepsilon_{\max}\left(\Omega\right) > 1 - \varepsilon$.

\section{Verification and information-preserving channel}\label{app:channel}

To construct the two-copy verification strategy, we consider
the case where the verifiers first implement the local operation and classical communication (LOCC) channel $\Lambda$. 
This channel treats the second copies as if they were an ideal graph state
and utilizes this entanglement resource to implement a series of non-local gates to the first copy. 
These gates are designed to perform unitary rotations, transforming an identical graph state into the specific state $\ket{0\cdots 0}$. Following this channel, everyone measures their first copies on the computational basis $\{\proj{0},\proj{1}\}$ and passed the test if the results are all $0$.
For simplicity, we use $\ket{\bm{0}}_n=\ket{0\cdots 0}$ to denote the basis state of $n$ -qubits. 
For the expected state $\ket{G}\ox\ket{G}$, it holds $\Lambda(\proj{G}\ox\proj{G})=\proj{\bm{0}}_n\ox\proj{\bm{0}}_n$.
To assess the efficiency for fake states, In this Appendix, 
we reformulate the optimization tasks in terms of information-preserving channels and establish the relation between channels 
and measurement operators as $\Omega_g = \Lambda^\dagger(\proj{\bm{0}}_n\ox\proj{\bm{0}}_n)$.
We need the following lemma, which follows directly from~\cite[Theorem 2]{Chitambar_2014}.

\begin{lemma}
    Any LOCC measurement strategy can be decomposed and consequently implemented through a LOCC channel within the same Hilbert space, followed by a measurement in the computational basis with a specific selection of binomial measurement results that yield the ``pass'' outcome.
\end{lemma}
One could set arbitrary binary string to the binomial measurement results with a ``pass'' outcome. However, the following theorem states that for a specific choice, $\{0\cdots0\}$, this strategy could formulate all the semi-optimal one-way strategies~\cite{Yu_2019}. 
\begin{theorem}
    Any semi-optimal one-way strategy~\cite{Yu_2019} can be constructed as 
    a one-way LOCC channel followed by a binomial passing choice represented as $0\cdots 0$.
\end{theorem}
\begin{proof}
    For a semi-optimal one-way strategy with target state $\ket{\psi}$, Alice chooses a measurement $\proj{v_i}$ with results $i=0,\cdots,n$. Subsequently, Bob performs measurements on $\proj{u_{t|i}}$ where $\ket{u_{0|i}} = \braket{v_i}{\psi}$. In accordance with this, we define a unitary matrix $U_i$ such that $\ket{0} = U_i \ket{u_{0|i}}$. The one-way LOCC channel can be expressed as
    \begin{align}
        \Lambda(\rho) &=\sum_i M_i \rho M_i^\dagger ,\\
        M_i &= \ket{0} \bra{v_i} \otimes U_i.
    \end{align}
    Subsequently, if Alice and Bob apply this channel first and then both measure on $\proj{0}, \proj{1}$ with the pass results represented by $\ket{00}$, they will get the same passing probability for any fake state $\sigma$.
\end{proof}
Thus we consider all the strategies that set the "pass" binomial measurement results as $\{0\cdots0\}$ and gives the channel optimization task below:
\begin{theorem}[Channel optimization] \label{theo:chanequ3}
  Fix the choice of "pass" binomial measurement results as $\{0\cdots 0\}$. Let's assume that $n$ independent parties share a state $\ket{\psi}$. A LOCC channel $\Lambda$ is optimal for verification if and only if it satisfies the following condition:

    \begin{enumerate}
        \item $\Lambda(\proj{\psi}) = \proj{0\cdots0}$.
        \item Any other LOCC channel $\Lambda'$ satisfied the first condition will cancel more information on the difference between $\sigma$ and $\ket{\psi}$ compared to $\Lambda$. In other words, 
\begin{align}
    \tr[\proj{\psi} \sigma] &\leq  \tr[\Lambda(\proj{\psi}) \Lambda(\sigma)] \\
    &\leq \tr[\Lambda'(\proj{\psi}) \Lambda'(\sigma)] .
\end{align}
    \end{enumerate}

This information-preserving channel, along with the verification strategy $\{\Omega,\mathbb{I} - \Omega\}$ constructed by this channel, then satisfies:
\begin{enumerate}
    \item $\Omega = \Lambda^\dagger(\proj{0\cdots0})$.
    \item $M_i\ket{\psi} = c_i \ket{0\cdots 0}$. Here $M_i$ is the Kraus operator of channel $\Lambda$, $c_i$ is a constant parameter.
    \item $\Lambda^\dagger(\proj{0\cdots0}) \ket{\psi} =\sum_i c_i M_i^\dagger \ket{0\cdots 0}= \ket{\psi}$.
\end{enumerate}

\end{theorem}

\begin{proof}
    We consider the strategy in which verifiers manipulate channel $\Lambda$ first and pass with all qubits in the result $\proj{0\cdots0}$. For any fake state $\sigma$, the passing probability is expressed as:
\begin{align}
    p(\Lambda) &= \bra{0\cdots0} \Lambda(\sigma) \ket{0\cdots 0}\\
    &=\tr[\Lambda(\sigma) \Lambda(\proj{\psi\cdots \psi})]  = \tr[\sigma \Lambda^\dagger(\proj{0\cdots0})] .
    \label{channelprob}
\end{align}
Consequently, we have derived the first conclusion that $\Omega = \Lambda^\dagger(\proj{0\cdots0})$. For any other channels, due to the second condition in the theorem, it must satisfy:
\begin{align}
    p(\Lambda') = \tr[\Lambda'(\sigma) \Lambda'(\proj{\psi\cdots \psi})] \geq p(\Lambda).
\end{align}
A larger passing probability implies that the $\Omega'$ generated by $\Lambda'$ will have less power to discern the fake state compared to $\Lambda$, showcasing the optimality of the channel construction within this fixed passing binomial choice. We suppose $M_i \ket{\psi} = c_i \ket{\Phi_i}$. Then the first condition becomes:
\begin{align}
    |c_i|^2 \sum_{i} \proj{\Phi_i} = \proj{0\cdots 0}.
\end{align}
Given that $\proj{0\cdots 0}$ is a pure state and lies at the boundary of the convex set, it implies that $\ket{\Phi_i} = \ket{0\cdots 0}$ must be satisfied, which proves the second conclusion. Regarding the last conclusion, we prove it with the calculations below:
\begin{align}
    \Lambda^\dagger(\proj{0\cdots0}) \ket{\psi} &= \sum_i M_i^\dagger \proj{0 \cdots 0} M_i \ket{\psi} \\
    &=\sum_i c_i M_i^\dagger \ket{0\cdots 0} \\
    &=\sum_i M_i^\dagger M_i \ket{\psi} = \ket{\psi}.
\end{align}
In the last equality, we use the trace one condition on the channel where $\sum_i M_i^\dagger M_i = \mathbb{I}$.
\end{proof}

\subsection{Demonstrative example: Bell state}

As a demonstrative example, 
we show that the single-copy optimal verification strategy 
for the Bell state~\cite{pallister2018optimal} can be reformulated in terms of quantum channels. 
Specifically, the optimal strategy has the following form~\cite{pallister2018optimal}
\begin{align}
    \Omega = \frac{1}{3} (P_{ZZ}^+ + P_{YY}^- + P_{XX}^+).
\end{align}
We construct the Karus operators according to this operator:
\begin{align}
    \Lambda(\rho) &= \sum_{i = 0}^{5} M_i \rho M_i^\dagger, \\
    M_0 &= \frac{1}{\sqrt{3}}\ket{0}\bra{0} \otimes \mathbb{I} \ , \ M_1 = \frac{1}{\sqrt{3}} \ket{0}\bra{1} \otimes X ,\\
    M_2 &= \frac{1}{\sqrt{3}}\ket{0}\bra{+} \otimes H \ , \ M_3 = \frac{1}{\sqrt{3}} \ket{0}\bra{-} \otimes XH ,  \\
     M_4 &= \frac{1}{\sqrt{3}}\ket{0}\bra{+i} \otimes S^* \ , \ M_5 = \frac{1}{\sqrt{3}} \ket{0}\bra{-i} \otimes X S^* .
\end{align}
It is easy to check that
$\sum_i M_i^\dagger M_i = \mathbb{I}$, $M_i \ket{\Phi} = \ket{00}/\sqrt{6}$, and
\begin{align}
    \Omega = \sum_{i = 0}^{i = 5} M_i^\dagger \proj{00} M_i \equiv \Lambda^\dagger(\proj{00}).
\end{align}
We observe that
\begin{align}
    p = \tr[\proj{00} \Lambda(\sigma)] = \tr[\Lambda(\proj{\psi})\Lambda(\sigma)] \geq \tr[\proj{\psi}\sigma].
\end{align}
The inequality is satisfied if and only if $\Lambda$ is a unitary channel. 
This represents the minimum passing probability for all measurement strategies 
and thus yields the globally optimal entangled 
measurement strategy $\{\proj{\psi}, \mathbb{I} -\proj{\psi}\}$, 
which may not always be realizable if only local operations and classical communication are allowed.

\section{Non-local gates through graph state entanglement}
\label{app:proofeqGraph} 
In this Appendix, we show that graph states can function as an entanglement resource to locally implement non-local control-$Z$ gates. In the following, we use $CZ_{AB}$ and $C_{AB}$ to denote control-$Z$ gate and control-$X$ gate with control qubit $A$ and target qubit $B$. We first prove the theorem below:
\begin{theorem}[Graph state disentangled equation]\label{Lemma:graph_disentangled}
Through local interactions $A_g$ between qubits $O_i'$ in graph state  $\ket{G}$ and auxiliary qubits $O_i$ in state $\ket{\omega}$, the graph state $\ket{G}$ associated with graph $g = (V,E)$ can operate as an entanglement resource, yielding a non-local unitary transformation $B_g$ on auxiliary qubits subject to local Pauli corrections denoted as $L_g(a)$ and $Q_g(a)$. Here $a$ is a binary string that represents different measurement results of qubits $O_i'$ on the computational basis. This non-local unitary matrix $B_g$ can transform one identical graph state to $\ket{0\cdots0}$:
\begin{align}    
A_g \ket{\omega}_O\otimes \ket{G}_{O'} &=\frac{1}{\sqrt{2^{|V|}}}\sum_{a} L_g(a) B_g  \ket{\omega}_O \otimes \ket{a}_{O'}, \label{eq:landr_graph_disentangled}
\end{align}
where 
\begin{align}   
A_g &= \prod_{i \in V } H_{O_i} C_{O_iO_i'}, \\
L_g(a) &= \prod_{(m,n) \in E} (-1)^{a_ma_n}X_{O_m} ^{a_n}X_{O_n}^{a_m}, \\
B_g &= \left(\prod_{i \in V} H_{O_i}\right) \times \left(\prod_{(m,n) \in E}CZ_{O_mO_n} \right).
\end{align}
Inversely, it holds that
\begin{align}
    A_g \ket{G}_{O} \otimes \ket{\omega}_{O'} 
=   \frac{1}{\sqrt{2^{|V|}}}\sum_{a} L_g(a)Q_g(a)B_g 
    \ket{\omega}_{O} \otimes \ket{a}_{O'} , \label{eq:randr_graph_disentangled}
\end{align}
where
\begin{align}
    Q_g(a) = \prod_{i \in V} Z_{O_i}^{a_i}.
\end{align}
\end{theorem}
\begin{proof}
To prove Theorem~\ref{Lemma:graph_disentangled}, it suffices to demonstrate the correctness of two disentangled equations below. 
\begin{align}
    \bra{a}_{O'} (\prod_{i \in V } H_{O_i} C_{O_iO_i'}) &\times (\prod_{(m,n) \in E} CZ_{O_m'O_n'} ) \ket{+}_{O'} \notag \\
    &=\frac{1}{\sqrt{2^{|V|}}} (\prod_{(m,n) \in E} (-1)^{a_ma_n}X_{O_m} ^{a_n}X_{O_n}^{a_m}) \times  (\prod_{i \in V} H_{O_i}) \times (\prod_{(m,n) \in E}CZ_{O_mO_n} ), \label{eq:rdisteq-graph-pre} \\
     \bra{a}_{O'} (\prod_{i \in V } H_{O_i} C_{O_iO_i'}) &\times (\prod_{(m,n) \in E} CZ_{O_mO_n} ) \ket{+}_{O} \notag \\
    &=\frac{1}{\sqrt{2^{|V|}}} (\prod_{(m,n) \in E} (-1)^{a_ma_n}X_{O_m} ^{a_n}X_{O_n}^{a_m}) \times  (\prod_{i \in V} Z_{O_i}^{a_i}) \times  (\prod_{i \in V} H_{O_i}) \times (\prod_{(m,n) \in E}CZ_{O_mO_n} ) \mathbb{I}_{O'O} .\label{eq:ldisteq-graph-pre}
\end{align}
We decomposed the state $\ket{\omega}$ in the computational basis: $\ket{\omega} = \sum_p \lambda_p \ket{p_0 \cdots p_N}$. For Eq.~\eqref{eq:rdisteq-graph-pre}, we calculate the expression below:
\begin{align}
   2^{|V|} \text{LHS~} \ket{\omega}_O &=\sqrt{2^{|V|}} \sum_{p,q} \bra{a}_{O'} \prod_{i \in V} H_{O_i} C_{O_iO_i'} \prod_{(m,n) \in E} (-1)^{q_mq_n} \lambda_p \ket{p_0 \cdots p_N}_O \otimes \ket{q_0 \cdots q_N}_{O'} \\
    &= \sum_{p,q,u} \lambda_p (-1)^{\sum_{i = 0}^N u_i p_i} \prod_{(m,n) \in E} (-1)^{q_mq_n}  \braket{a}{q_0 + p_0,\cdots,q_N+p_N} \ket{u_0,\cdots,u_N}_O \\
    &= \sum_{p,u} \lambda_p (-1)^{\sum_{i = 0}^N u_i p_i} \prod_{(m,n) \in E} (-1)^{(a_m+p_m)(a_n+p_n)}  \ket{u_0,\cdots,u_N}_O \\
    &= \sum_{p,u} \lambda_p (-1)^{\sum_{i = 0}^N u_i p_i} \prod_{(m,n) \in E} (-1)^{a_m \cdot a_n} \prod_{(m,n) \in E} (-1)^{p_m \cdot p_n} \prod_{(m,n) \in E} (-1)^{a_m \cdot p_n + p_m \cdot a_n}  \ket{u_0,\cdots,u_N}_O .
\end{align}
\begin{align}
   2^{|V|} \text{RHS~}  \ket{\omega}_O&= \sqrt{2^{|V|}} \sum_p \lambda_p \prod_{(m,n) \in E} (-1)^{a_ma_n} ( \prod_{(m,n) \in E} X_{O_m}^{a_n} X_{O_n}^{a_m}) (\prod_{i \in V} H_{O_i}) \times (\prod_{(m,n) \in E}CZ_{O_mO_n} ) \ket{p_0\cdots p_N}_O \\
   &=\sum_{u,p} \lambda_p \prod_{(m,n) \in E} (-1)^{a_ma_n} ( \prod_{(m,n) \in E} X_{O_m}^{a_n} X_{O_n}^{a_m})  (-1)^{\sum_{i = 0}^N u_i p_i} \times \prod_{(m,n) \in E}(-1)^{p_mp_n}  \ket{u_0\cdots u_N}_O \\
   &=\sum_{u,p} \lambda_p (-1)^{\sum_{i = 0}^N u_i p_i}  \prod_{(m,n) \in E} (-1)^{a_ma_n}\prod_{(m,n) \in E}(-1)^{p_mp_n}  ( \prod_{(m,n) \in E} X_{O_m}^{a_n} X_{O_n}^{a_m})   \ket{u_0\cdots u_N}_O \\
    &= \sum_{p,u'} \lambda_p (-1)^{\sum_{i = 0}^N u_i' p_i} \prod_{(m,n) \in E} (-1)^{a_m \cdot a_n} \prod_{(m,n) \in E} (-1)^{p_m \cdot p_n} \prod_{(m,n) \in E} (-1)^{a_m \cdot p_n + p_m \cdot a_n}  \ket{u_0',\cdots,u_N'}_O.
\end{align}
This coincidence proves Eq.~\eqref{eq:rdisteq-graph-pre}.
For Eq.~\eqref{eq:ldisteq-graph-pre}, the same calculation proceeds as follows:
\begin{align}
   2^{|V|} \text{LHS~} \ket{\omega}_{O'}&=\sqrt{2^{|V|}} \sum_{p,q} \bra{a}' \prod_{i \in V} H_{O_i} C_{O_iO_i'} \prod_{(m,n) \in E} (-1)^{q_mq_n} \lambda_p \ket{q_0 \cdots q_N}_O \otimes \ket{p_0 \cdots p_N}_{O'} \\
    &= \sum_{p,q,u} \lambda_p (-1)^{\sum_{i = 0}^N u_i q_i} \prod_{(m,n) \in E} (-1)^{q_mq_n}  \braket{a}{q_0 + p_0,\cdots,q_N+p_N} \ket{u_0,\cdots,u_N}_{O} \\
    &= \sum_{p,u} \lambda_p (-1)^{\sum_{i = 0}^N u_i (a_i+p_i)} \prod_{(m,n) \in E} (-1)^{(a_m+p_m)(a_n+p_n)}  \ket{u_0,\cdots,u_N}_O \\
    &= \sum_{p,u} \lambda_p (-1)^{\sum_{i = 0}^N u_i p_i} \!\!\!\! \!\!\!\!  \prod_{(m,n) \in E}\!\!\!\!  (-1)^{a_m \cdot a_n} \!\!\!\! \!\!\!\! \prod_{(m,n) \in E} \!\!\!\! (-1)^{p_m \cdot p_n} (-1)^{\sum_{i = 0}^N u_i a_i}\!\!\!\! \!\!\!\! \prod_{(m,n) \in E} \!\!\!\! (-1)^{a_m \cdot p_n + p_m \cdot a_n}  \ket{u_0,\cdots,u_N}_O .
\end{align}
\begin{align}
   2^{|V|} \text{RHS~}  \ket{\omega}_{O'}
   &=\sum_{u,p} \lambda_p \!\!\!\! \prod_{(m,n) \in E}\!\!\!\!  (-1)^{a_ma_n} (\!\!\!\!  \prod_{(m,n) \in E}\!\!\!\!  X_{O_m}^{a_n} X_{O_n}^{a_m}) (-1)^{\sum_{i = 0}^N u_i a_i} (-1)^{\sum_{i = 0}^N u_i p_i} \times \!\!\!\! \prod_{(m,n) \in E}\!\!\!\! (-1)^{p_mp_n}  \ket{u_0\cdots u_N}_O \\
   &=\sum_{u,p} \lambda_p (-1)^{\sum_{i = 0}^N u_i p_i}  \!\!\!\! \prod_{(m,n) \in E} \!\!\!\! (-1)^{a_ma_n}\!\!\!\! \prod_{(m,n) \in E}\!\!\!\! (-1)^{p_mp_n}   (-1)^{\sum_{i = 0}^N u_i a_i} ( \!\!\!\! \prod_{(m,n) \in E}\!\!\!\!  X_{O_m}^{a_n} X_{O_n}^{a_m})   \ket{u_0\cdots u_N}_O \\
    &= \sum_{p,u'} \lambda_p (-1)^{\sum_{i = 0}^N u_i' p_i} \!\!\!\! \!\!\!\!  \prod_{(m,n) \in E}\!\!\!\!  (-1)^{a_m \cdot a_n} \!\!\!\! \!\!\!\! \prod_{(m,n) \in E}\!\!\!\!  (-1)^{p_m \cdot p_n}  (-1)^{\sum_{i = 0}^N u_i' a_i} \!\!\!\! \prod_{(m,n) \in E}\!\!\!\!  (-1)^{a_m \cdot p_n + p_m \cdot a_n}  \ket{u_0',\cdots,u_N'}_O.
\end{align}
This proves the Eq.~\eqref{eq:ldisteq-graph-pre}.  
\end{proof}
It is noteworthy that $A_g$ represents local operations with respect to different verifiers because two-qubit gates $C_{O_iO_i'}$ can be locally realized in each verifier's quantum memory. However, the operator $B_g$ in Theorem~\ref{Lemma:graph_disentangled} involves non-local operations-Control-Z gates between qubits at different parties. This nonlocality arises from the consumption of the entanglement resource of the graph state $\ket{G}$. 

Numerous pertinent observations merit discussion. 

Firstly, in the context of two qubits, the equation presented in Theorem~\ref{Lemma:graph_disentangled} converges to the optimal local implementation of the CNOT gate, as elucidated in prior research~\cite{PhysRevA.62.052317}. Consequently, we anticipate that a singular entangled graph state, coupled with adjacent edge communication, will prove efficacious for the local implementation of control-Z gates between graph edges. This enables the restoration of a set of control-$Z$ gates in the entangled state, which could be generated through Ising interactions, and facilitates rapid implementation of those gates on multiple remote computers via entanglement distribution and local gates. This methodology holds particular promise for applications in distributed quantum computation.
\begin{figure}[!htbp]
    \centering
    \includegraphics[width=0.8\linewidth]{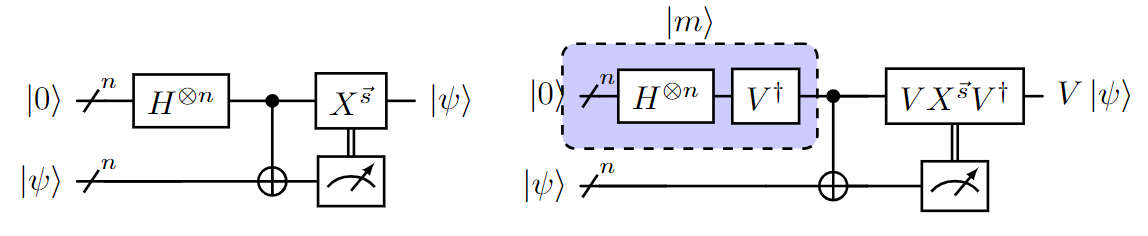}
    \caption{The one-bit gate teleportation circuit guarantees that \( VX^{\vec{s}}V^{\dagger} \) remains a local gate for any \(\vec{s}\), and that \( V^\dagger \) commutes with \(\text{CNOT}\) gates. Consequently, \( V^\dagger \) must be a Clifford gate and a phase gate, modulo some local unitary gates.} 
    \label{fig:teleportation}
\end{figure}

Secondly, how to use quantum state as resource to generate certain quantum gate has been discussed in the region of fault-tolerant quantum computation.
Proof above can be well understood in the one-bit teleportation construction in Figure~\ref{fig:teleportation}~\cite{PhysRevA.62.052316}, where $\ket{m}$ represents our graph states and $V$ is the gate it locally implements. 
Any phase gate $V$ can be decomposed into a composition of $\{U, CZ, CCZ, \cdots\}$.
However, any presence of $C^kZ$ will lead to a $C^{k-1}Z$ correction on $VX^{\vec{s}}V^{\dagger}$,
thus turning the verification strategy into a non-local one.
We can then conclude that only states equivalent to graph states under local unitary transformations can be used as a non-local gate implementation resource.
Given that any stabilizer state can be transformed into a graph state through local Clifford (LC) operations~\cite{PhysRevA.69.022316},
our strategy can be generalized to all stabilizer states.

 It is valuable to determine how to locally use non-stabilizer states as a "gate resource" in the future research.
One might use the recursive construction of teleportation circuits mentioned by Zhou \emph{et al.} \cite{PhysRevA.62.052316}. 
Exploring qudit teleportation circuits might be another direction~\cite{de2021efficient}.

Lastly, the presented theorem introduces an efficient two-copy verification strategy 
for graph states, as shown in the next section.

\section{Two-copy verification for graph states}\label{app:graph_disentangled}
In the two-copy case, the optimization tasks can also be reframed from the channel perspective. For simplicity, we use $\ket{\textbf{0}}$ to denote the vector $\ket{0\cdots0}$ in the first or second copy space. We consider the channel that regards $\proj{\psi}$ as an entanglement resource to implement a nonlocal unitary transformation on the second copy, this unitary transformation rotates $\ket{\psi}$ to $\ket{0\cdots 0}$ and $\ket{\psi^\perp}$ to another basis orthogonal to $\ket{0\cdots0}$. To formulate a $(n, d, 2)$ strategy for general graph states, we begin by constructing the Kraus operators, outlined below:
\begin{align}
\Lambda(\rho) &= \sum_{i=b_1\cdots b_{n}} M_i \rho M_i^\dagger  ,\notag\\
    M_{b_1\cdots b_n} &=[L_g^\dagger(\bm{b}) \otimes \ket{\textbf{0}}_n\bra{\bm{b}} ] \times A_g.
\end{align}
Consequently, the measurement operator takes the following form according to Theorem~\ref{theo:chanequ3}.
\begin{align}
     \Omega_g &=\sum_{\bm{b}} A_g^\dagger [L_g(\bm{b}) \proj{\textbf{0}}_n L_{g}^\dagger(\bm{b})] \otimes \proj{\bm{b}}  A_g \\
    &= \sum_{b_1,\cdots,b_{n} \in \{0,1\}} A_g^\dagger \bigotimes_{i = 1}^{n} [\proj{\cP(\!\!\!\!\!\! \sum_{\substack{j \\ (b_j,b_i) \in E}} \!\!\!\!\!\!b_j)}_{O_i} \otimes \proj{b_i}_{O_i'}] A_g \\
    & = \sum_{b_1,\cdots,b_{n} \in \{0,1\}} A_g^\dagger \bigotimes_{i = 0}^{N-1} [\proj{c_i(\bm{b})}_{O_i} \otimes \proj{b_i}_{O_i'}] A_g. 
\end{align}
Here, $\cP$ denotes a parity projection on $0$ or $1$, and $c_i(\bm{b})$ is a newly generated string according to string $\bm{b} = (b_1,\cdots,b_n)$ and graph $(V,E)$. We term $c(\bm{b})$ the \textit{graph parity string} of $\bm{b}$ with respect to graph $G$, indicating that $c_i(\bm{b})$ at a specific vertex, $i$ corresponds to the parity projection of the summation of all $b_j$ at the adjacent vertices. It is worth noting that $A_g^\dagger$ represents a basis transformation from the computational basis ${\ket{00},\ket{01},\ket{10},\ket{11}}_{OO'}$ to four maximally entangled states:
\begin{align}
\ket{\Phi_{00}}_{OO'} = \frac{\ket{00}+\ket{11}}{\sqrt{2}},\ \ket{\Phi_{01}}_{OO'} = \frac{\ket{01}+\ket{10}}{\sqrt{2}},\ \ket{\Phi_{10}}_{OO'} = \frac{\ket{00}-\ket{11}}{\sqrt{2}},\ \ket{\Phi_{11}}_{OO'} = \frac{\ket{01}-\ket{10}}{\sqrt{2}}.
\end{align}
The corresponding measurement operator is then given by:
\begin{align}
    \Omega_g = \sum_{\bm{b} \in \{0,1\}^n} \bigotimes_{j = 1}^{n} \proj{\Phi_{c_j(\bm{b}) b_j}}_{O_jO_j'}.
\end{align}
This operator indeed satisfy the symmetric condition of Eq.~\eqref{eq:critical2} 
for the symmertic property of basis ${\ket{\Phi_{00}},\ket{\Phi_{01}},\ket{\Phi_{10}},\ket{\Phi_{11}}}$ themselves.
What's more, we can prove that $\Omega_g$ already attained the global-optimal lower bounds.
To execute the strategy $\Omega_g$ , each party can measure their two qubits in the entanglement basis ${\ket{\Phi_{00}},\ket{\Phi_{01}},\ket{\Phi_{10}},\ket{\Phi_{11}}}$ with corresponding measurement outcomes ${00,01,10,11}$. Subsequently, they separate their first and second digits to form the strings $a$ and $b$, respectively, and verify whether $b$ constitutes the graph parity string of $a$ with respect to graph $G$. 

If we define the state $\ket{s_{inv}} = \ket{G}\otimes\ket{\omega}$, applying Eq.~\eqref{eq:randr_graph_disentangled}, we can deduce:
\begin{align}
  M_a \ket{s_{inv}} 
= \frac{1}{\sqrt{2^{|V|}}} (Q_g(a) \times B_g)\otimes \mathbb{I} \ket{\omega} \otimes\ket{0\cdots 0}
= \frac{1}{\sqrt{2^{|V|}}}V(a)\otimes \mathbb{I} \ket{\omega} \otimes\ket{0\cdots 0}.
\end{align}
While the unitary operator $V(a)$ is associated with $a$, it consistently rotates the graph state into $\ket{0\cdots 0}$ due to the stabilizing property of $Q(a)$ for the state $\ket{0\cdots 0}$. In this case, we calculate that:
\begin{align}
     \Omega_g \ket{s_{inv}}= \Lambda^\dagger (\proj{\textbf{0}}_n\otimes \proj{\textbf{0}}_n) \ket{s_{inv}} &= \sum_i M_i^\dagger \proj{\textbf{0}}_n\otimes \proj{\textbf{0}}_n M_i \ket{s_{inv}} \\
    &= \sum_i M_i^\dagger \ket{\textbf{0}}_n \otimes \ket{\textbf{0}}_n \times[\bra{\textbf{0}}\frac{1}{\sqrt{2^{|V|}}}V(a) \ket{\omega}] \\
    &= \braket{G}{\omega}\sum_i \frac{1}{\sqrt{2^{|V|}}} M_i^\dagger \ket{\textbf{0}}_n \otimes \ket{\textbf{0}}_n  \\
    &= \braket{G}{\omega} \times \ket{G} \otimes \ket{G}.
\end{align}
The second equality relies on the fact that $V(a)$ unitarily transforms $\ket{\psi}$ to $\ket{\textbf{0}}_n$. If we define the state $\ket{s} = \ket{\omega}\otimes\ket{G}$, where $\ket{G}$ is the graph state, employing Eq.~\eqref{eq:landr_graph_disentangled}, it becomes evident that:
\begin{align}
    M_a \ket{s} &=\frac{1}{\sqrt{2^{|V|}}} B_g \ket{\omega} \otimes\ket{0\cdots 0}.
\end{align}
Here, $B_g$ represents a unitary transformation that maps the graph state to $\ket{0\cdots 0}$. Similarly

\begin{align}
    \Omega_g \ket{s}= \Lambda^\dagger (\proj{\textbf{0}}_n\otimes \proj{\textbf{0}}_n) \ket{s} &= \sum_i M_i^\dagger \proj{\textbf{0}}_n\otimes \proj{\textbf{0}}_n M_i \ket{s} \\
    &= \sum_i M_i^\dagger \ket{\textbf{0}}_n \otimes \ket{\textbf{0}}_n [\bra{\textbf{0}}\frac{1}{\sqrt{2^{|V|}}}B_g \ket{\omega}] \\
    &= \braket{G}{\omega}\sum_i \frac{1}{\sqrt{2^{|V|}}} M_i^\dagger \ket{\textbf{0}}_n \otimes \ket{\textbf{0}}_n  \\
    &= \braket{G}{\omega} \ket{G} \otimes \ket{G}'.
\end{align}
We refer to $\Lambda$ as the self-disentangled channel of state $\ket{\psi}$. If we substitute $\ket{\omega} = \ket{\psi} = \ket{G}$, we will observe that Eq.~\eqref{eq:critical1} is satisfied. If we substitute $\ket{\omega} = \ket{\psi^\perp}$, we see that:
\begin{align}
    \Omega_g  \mathbb{P}_s \ket{\psi\psi^\perp} =\frac{1}{2} \Omega_g  (\ket{\psi\psi^\perp} + \ket{\psi^\perp\psi}) = 0.
\end{align}

Now, we can calculate the operator $\Omega_{\star}$ in Theorem~\ref{theorem:optimization_target-app}. 
First, we demonstrate that $\Omega_g \mathbb{P}_s \mathbb{P}_\psi = 0$:
\begin{align}
  \Omega_g \mathbb{P}_s\mathbb{P}_\psi 
= \Omega_g \sum_j \mathbb{P}_s \proj{\psi\psi^\perp_j} 
= \sum_j \braket{\psi}{\psi^\perp_j} \ket{\psi \psi} \bra{\psi\psi^\perp_j} 
= 0.
\end{align}
Then we conclude that  $ \Omega_{\star} = 2\mathbb{P}_\psi\mathbb{P}_s\Omega_g\mathbb{P}_s \mathbb{P}_\psi = 0$.  Such a strategy corresponds to a scenario where $\lambda_{\star}(\Omega_g) = 0$. 
Similarly, we can show that $\Omega_g \mathbb{P}_{\psi} = 0 $ and conclude that $\xi_{\star}(\Omega_g) = \gamma_{\star}(\Omega_g) = 0$. This means that $\varepsilon_{\max}\left(\Omega\right) \gg \varepsilon$.  
Thus, this implies that the verification efficiency is:
\begin{align}
    N_m(\Omega_g) = \frac{1}{\varepsilon (1 - \lambda_{\star}(\Omega_g)) + O(\varepsilon^{1.5})} \ln \frac{1}{\delta}= \frac{1}{\varepsilon + O(\varepsilon^{1.5}) } \ln \frac{1}{\delta}.
\end{align}
When $\varepsilon$ is sufficiently small, the efficiency converges to the globally optimal efficiency. 

We can directly calculate the passing probability of strategy $\Omega_g$ as
\begin{align}\label{eq:pass_prob_graph}
    p(\Omega_g)
= \bra{\sigma\sigma'}\Omega_g\ket{\sigma\sigma'}= (1 - \varepsilon_r)(1 - \varepsilon_r') + \varepsilon_r \varepsilon_r' \bra{\psi^\perp\psi'^\perp}\Omega_g\ket{\psi^\perp\psi'^\perp}
\approx 1 - \varepsilon_r - \varepsilon_r'.
\end{align}
This equation also illustrates that $\Omega_g$ is already a global-optimal strategy itself. 
To evaluate $\varepsilon_{\max}\left(\Omega\right)$, we note that:
\begin{align}
    \frac{dp}{d \varepsilon_r} = -1 + \varepsilon_r'+\varepsilon_r'\bra{\psi^\perp\psi'^\perp}\Omega_g\ket{\psi^\perp\psi'^\perp},\\
    \frac{dp}{d \varepsilon_r'} = -1 + \varepsilon_r +\varepsilon_r\bra{\psi^\perp\psi'^\perp}\Omega_g\ket{\psi^\perp\psi'^\perp}.
\end{align}
Then the
\begin{align}
    \frac{dp}{d\varepsilon_r} -\frac{dp}{d\varepsilon_r'} = -(1 + \bra{\psi^\perp\psi'^\perp}\Omega_g\ket{\psi^\perp\psi'^\perp})(\varepsilon_r - \varepsilon_r').
\end{align}
This implies that $p$ reaches its maximum when $\varepsilon_r = \varepsilon_r'$. 
Subsequently, we simplify the function as follows:
\begin{align}
    p(\varepsilon_r) = 1 - 2\varepsilon_r + \varepsilon_r^2 [1 + \bra{\psi^\perp\psi'^\perp}\Omega_g\ket{\psi^\perp\psi'^\perp}],
\end{align}
To confirm that this function reaches the maximum at $\varepsilon_r = \varepsilon$, we find another solution $\varepsilon_{\max}\left(\Omega\right)$ such that $p(\varepsilon) = p(\varepsilon_{\max}\left(\Omega\right))$ is satisfied:
\begin{align}
    \varepsilon_{\max}\left(\Omega\right) 
= \frac{2}{1 + \bra{\psi^\perp\psi'^\perp}\Omega_g\ket{\psi^\perp\psi'^\perp}} 
- \varepsilon > 1 - \varepsilon.
\end{align}
This provides a lower bound for $\varepsilon_{\max}\left(\Omega\right)$. For a sufficiently small $\varepsilon$, it is effective to verify that $\varepsilon_{r} > \varepsilon_{\max}\left(\Omega\right)$. 

\subsection{Demonstrative example: Two-copy Bell state verification} 

Here, we consider a straightforward scenario where two copies of the Bell states $O_0,O_1$ and $O_0',O_1'$ are distributed to parties $0$ and $1$. A Bell state is equivalent to a graph state up to a unitary transformation $H_{O_0}H_{O_0'}$ applied to the Bell state. Thus, the Bell state can be efficiently verified by the two-copy verification operator, as illustrated in Table~\ref{tab:LGS2_all_parity_string}:
\begin{align}
    \Omega_g 
&= \proj{\Phi_{00}}_{O_0O_0'} \otimes \proj{\Phi_{00}}_{O_1O_1'} + \proj{\Phi_{11}}_{O_0O_0'} \otimes \proj{\Phi_{11}}_{O_1O_1'} \notag \\
&\qquad + \proj{\Phi_{01}}_{O_0O_0'} \otimes \proj{\Phi_{10}}_{O_1O_1'} + \proj{\Phi_{10}}_{O_0O_0'} \otimes \proj{\Phi_{01}}_{O_1O_1'}.
\end{align}
\begin{table}[!hbtp]
\centering
\setlength{\tabcolsep}{8pt}
\setlength\heavyrulewidth{0.3ex}
\renewcommand{\arraystretch}{1.5}
\begin{tabular}{@{}ccc@{}}
    \toprule
        \textbf{ Code $(a_0,a_1)$} & \textbf{Code $(b_0,b_1)$}& \textbf{Passing Measurement} \\\hline
        $(0,0)$ & $(0,0)$& $ \proj{\Phi_{00}} \otimes \proj{\Phi_{00}}$ \\
       $(0,1)$  & $(1,0)$& $ \proj{\Phi_{10}} \otimes \proj{\Phi_{01}}$ \\
        $(1,0)$ & $(0,1)$& $ \proj{\Phi_{01}} \otimes \proj{\Phi_{10}}$ \\
         $(1,1)$ & $(1,1)$& $ \proj{\Phi_{11}} \otimes \proj{\Phi_{11}}$ \\
    \bottomrule
\end{tabular}
\caption{The table presents all four graph parity codes for a two-qubit linear graph state. 
Each parity code corresponds to a projective measurement on the Bell basis. 
The weighted sum of these passing measurement operators lead to our strategy operator $\Omega_g$.}
\label{tab:LGS2_all_parity_string}
\end{table}

Furthermore, we notice that
\begin{align}
 H_{O_0}H_{O_0'} \ket{\Phi_{01}}_{O_0O_0'} = \ket{\Phi_{10}}_{O_0O_0'}, \\
 H_{O_0}H_{O_0'} \ket{\Phi_{10}}_{O_0O_0'} = \ket{\Phi_{01}}_{O_0O_0'} ,\\
  H_{O_0}H_{O_0'} \ket{\Phi_{00}}_{O_0O_0'} = \ket{\Phi_{00}}_{O_0O_0'} ,\\
   H_{O_0}H_{O_0'} \ket{\Phi_{11}}_{O_0O_0'} = \ket{\Phi_{11}}_{O_0O_0'}.
\end{align}
Therefore, the two-copy verification strategy of Bell state is given by:
  \begin{align}
    \Omega_{\rm Bell} 
&= H_{O_0}H_{O_0'} \Omega_g H_{O_0}H_{O_0'}\\
&= \proj{\Phi_{00}}_{O_0O_0'} \otimes \proj{\Phi_{00}}_{O_1O_1'} + \proj{\Phi_{11}}_{O_0O_0'} \otimes \proj{\Phi_{11}}_{O_1O_1'} \notag \\
&\quad + \proj{\Phi_{01}}_{O_0O_0'} \otimes \proj{\Phi_{01}}_{O_1O_1'} + \proj{\Phi_{10}}_{O_0O_0'} \otimes \proj{\Phi_{10}}_{O_1O_1'}.
\end{align}
It should be noted that this measurement operator in the $16 \times 16$ linear space 
may possess multiple unit eigenvalues. 
To numerically gauge its efficiency, we perform the Schmidt decomposition of the operator:
 \begin{align}
    \Omega_{\rm Bell} - \proj{\psi}\otimes \proj{\psi} &= \sum_{i} \Lambda_i M_{O_0,O_1}^{(i)} \otimes M_{O_0',O_1'}^{(i)}  \\
    &= \begin{pmatrix}
    0.5&0&0&-0.5\\
    0&0 &0&0\\
    0&0&0 &0\\
    -0.5&0&0&0.5\\
    \end{pmatrix}_{O_0,O_1} \otimes 
    \begin{pmatrix}
    0.5&0&0&-0.5\\
    0&0 &0&0\\
    0&0&0 &0\\
    -0.5&0&0&0.5\\
    \end{pmatrix}_{O_0',O_1'} \notag \\
    &\qquad+ \begin{pmatrix}
    0&0&0&0\\
    0&0 &\frac{1}{\sqrt{2}}&0\\
    0&\frac{1}{\sqrt{2}}&0 &0\\
    0&0&0&0\\
    \end{pmatrix}_{O_0,O_1} \otimes 
    \begin{pmatrix}
    0&0&0&0\\
    0&0 &\frac{1}{\sqrt{2}}&0\\
    0&\frac{1}{\sqrt{2}}&0 &0\\
    0&0&0&0\\
    \end{pmatrix}_{O_0',O_1'} \notag \\
    &\qquad+ \begin{pmatrix}
    0&0&0&0\\
    0&\frac{1}{\sqrt{2}} &0&0\\
    0&0&\frac{1}{\sqrt{2}} &0\\
    0&0&0&0\\
    \end{pmatrix}_{O_0,O_1}\otimes 
    \begin{pmatrix}
    0&0&0&0\\
    0&\frac{1}{\sqrt{2}} &0&0\\
    0&0&\frac{1}{\sqrt{2}} &0\\
    0&0&0&0\\
    \end{pmatrix}_{O_0',O_1'}.
\end{align}
We observe that all three matrices satisfy $M^{(i)} \ket{\psi} = 0$. This indicates that the term $\Omega \ket{\psi} \otimes \ket{\psi^{\perp}}$ will vanish during the calculation. Now, let's consider the fake state $\ket{\sigma} \otimes \ket{\sigma'}$ according to Lemma~\ref{lemma:pure-states} and calculate:
\begin{align}
    p(\Omega_{\rm Bell}, \ket{\sigma,\sigma'}) 
&= |\braket{\psi}{\sigma}\braket{\psi}{\sigma'}| + \bra{\sigma\sigma'}\left(\sum_{i} 
    \Lambda_i M_{O_0,O_1}^{(i)} \otimes M_{O_0',O_1'}^{(i)}\right)\ket{\sigma\sigma'} \\
&= (1 - \varepsilon_r)\times (1 - \varepsilon_r') + \varepsilon_r \varepsilon_r' 
    \bra{\psi^\perp \psi'^\perp}
    \left(\sum_{i} \Lambda_i M_{O_0,O_1}^{(i)} \otimes M_{O_0',O_1'}^{(i)}\right)
    \ket{\psi^\perp \psi'^\perp}  \\
    &\leq (1 - \varepsilon_r)\times (1 - \varepsilon_r') + \varepsilon_r \varepsilon_r'.
\end{align}
When we only consider the fake state near the target state, where $\varepsilon_r$ is sufficiently small, the linear term predominates. In this case, we can conclude that
\begin{align}
  p(\Omega_{\rm Bell},\ket{\sigma,\sigma'}) \leq (1 - \varepsilon_r)\times (1 - \varepsilon_r') + \varepsilon_r \varepsilon_r' \leq (1 - \varepsilon)^2 + \varepsilon^2 . 
\end{align}
The condition that the fake state is near the target state is critical. Specifically, we observe that setting $\varepsilon_r = \varepsilon_r' = 1$ leads to $p(\Omega,\ket{\sigma,\sigma'}) = 1$, which is certainly greater than $(1 - \varepsilon)^2 + \varepsilon^2$. However, in this extreme case, the fake state $\ket{\sigma\sigma'}$ is already orthogonal to the target state $\ket{\psi}\otimes \ket{\psi}$. Therefore, this fake state can be easily verified with a standard single-copy strategy. After obtaining $p(\Omega_{\rm Bell})$, 
we can draw conclusions according to Eq.~\eqref{eq:Nm}:
\begin{align}
   N_m(\Omega_{\rm Bell}) 
&= \frac{2\ln\delta}{\ln\left[(1 -2\varepsilon + \varepsilon^2) + \varepsilon^2\!\right]}.
\end{align}
When both $\varepsilon$ and $\delta$ are sufficiently small, the required number of copies scales as $1/\varepsilon\ln1/\delta$, approaching the optimal strategy.

\subsection{Application: Graph state Fidelity estimation}\label{appx:fidelity-estimation}

For quantum devices that generate quantum states with independent and identical distribution, we can regard the fake state as two copies of  $\sigma \otimes\sigma'$.
Supposed that:
\begin{align}
    \sigma = \sum_i p_i \proj{\sigma_i},\quad \sigma' = \sum_j p_j' \proj{\sigma_j'}.
\end{align}
Here $\ket{\sigma_i}$ and $\ket{\sigma_j'}$ are both pure states that satisfy:
\begin{align}
    \ket{\sigma_i} = \sqrt{1 - \varepsilon_{ri}} \ket{G} + \sqrt{\varepsilon_{ri}} \ket{G_i^\perp}, \\
    \ket{\sigma_j'} = \sqrt{1 - \varepsilon'_{rj}} \ket{G} + \sqrt{\varepsilon'_{rj}} \ket{G_j'^{\perp}}.
\end{align}
Then the fidelity of these two copies satisfied:
\begin{align}
    F := \bra{G} \sigma \ket{G} = \sum_i p_i (1 - \varepsilon_{ri}), \quad F' := \bra{G} \sigma' \ket{G} = \sum_j p_j' (1 - \varepsilon_{rj}').
\end{align}
The passing rate of strategy $\Omega_g$ can be evaluated via
\begin{align}
  p_s 
= \tr[\Omega_g(\sigma \otimes \sigma')] 
= \sum_i \sum_j p_i p_j' \bra{\sigma_i \sigma_j'}\Omega_g\ket{\sigma_i \sigma_j'}.
\end{align}
According to Eq.~\eqref{eq:pass_prob_graph}, 
\begin{align}
    \bra{\sigma_i\sigma_j'}\Omega_g\ket{\sigma_i\sigma_j'}= (1 - \varepsilon_{ri})(1 - \varepsilon_{rj}') + A_{ij}\varepsilon_{ri}\varepsilon_{rj}',
\end{align}
where $A_{ij} := \bra{G_i^\perp G'^\perp_j}\Omega_g\ket{G_i^\perp G'^\perp_j}$. 
Then $p_s$ satisfies:
\begin{align}
    p_s &= \sum_i \sum_j p_i p_j' [(1 - \varepsilon_{ri})(1 - \varepsilon_{rj}') + A_{ij}\varepsilon_{ri}\varepsilon_{rj}'] \\
    &=\left[\sum_i p_i(1 - \varepsilon_{ri})\right]\times \left[\sum_j p_j'(1 - \varepsilon_{rj}')\right] + \sum_i \sum_j p_i p_j' A_{ij}\varepsilon_{ri}\varepsilon_{rj}' \\
    &= F \times F' + \sum_i \sum_j p_i p_j' A_{ij}\varepsilon_{ri}\varepsilon_{rj}'.
\end{align}
When the infidelity $\varepsilon$ is sufficiently small, we have $(1 - F) \sim \cO(\varepsilon)$ and $(1 - F') \sim \cO(\varepsilon)$. 
Because $p_i,p_j' \in [0,1]$, it holds that for all $i$ and $j$,
\begin{align}
    p_i\varepsilon_{ri} \sim \cO(\varepsilon) ,\quad p_j'\varepsilon_{rj}' \sim \cO(\varepsilon).
\end{align}
Thus, we can conclude that
\begin{align}
p_s = \bra{G} \sigma \ket{G} \cdot \bra{G} \sigma' \ket{G} + \cO(\varepsilon^2).
\end{align}

\section{Illustrative example for dimension expansion strategy}\label{appx:dimension-expansion}

In this section, we present the explicit dimension expansion construction of a $(2,2,2)$-strategy applicable to the two-qubit state $\ket{\psi_\theta}=\cos\theta\ket{00}_{AB}+\sin\theta\ket{11}_{AB}$ distributed between two parties, 
denoted as $A$ and $B$, representing a specific instance of GHZ-like states. A more generalized approach can be formulated similarly, drawing from the efficient $(n,1,d)$ verification strategy proposed by Li \emph{et al.} 
for GHZ-like qudit states~\cite{Li_2020}.
According to the dimension expansion method outlined in the main text, this task is equivalent to identifying a $(2,1,2^2)$ strategy applicable to the GHZ-like qudit state of the form:
\begin{align}
    \ket{\Psi_\theta} 
:= \cos^2 \! \theta \ket{\textbf{0}}_A \otimes \ket{\textbf{0}}_B 
+ \cos\theta \sin\theta (\ket{\textbf{1}}_A \otimes \ket{\textbf{1}}_B 
+ \ket{\textbf{2}}_A \otimes \ket{\textbf{2}}_B) 
+ \sin^2 \! \theta \ket{\textbf{3}}_A \otimes \ket{\textbf{3}}_B.
\end{align}
This state corresponds to the $2$-th tensor product state $\ket{\psi_{ \theta}}^{\otimes 2}$ via the following identification: $\ket{00}\to\ket{\textbf{0}}, \ket{01}\to\ket{\textbf{1}},\ket{10}\to\ket{\textbf{2}},\ket{11}\to\ket{\textbf{3}}$.

In a previous study, various straightforward and efficient protocols were proposed for verifying bipartite 
qudit pure states such as $\ket{\Psi_\theta}$. 
Here, we employ a strategy based on a two-way LOCC (Local Operations and Classical Communication) strategy, 
denoted as $\Omega_{\rm IV}$ in~\cite[Eq. (48)]{Li_2020}. 
It is worth noting that alternative qudit measurement strategies offering higher efficiency can be chosen, 
leading to a different $(2,2,2)$ multi-copy strategy requiring fewer copies.

In the qudit verification strategy $\Omega_{\rm IV}$, five measurement bases are initially defined based on the five mutually unbiased bases (MUBs) for the four-dimensional space. Through the identification between qubit states and qudit states, these bases can be explicitly expressed as follows. It is important to note that the projective measurements onto the latter four sets of bases necessitate local interactions between two qubits in each party.
\begin{subequations}
\begin{align*}
    B_0 &: \left\{\ket{00},\ket{01},\ket{10},\ket{11}\right\} ,\\
    B_1 &: \left\{\frac{\ket{00}+\ket{01}+\ket{10}+\ket{11}}{2},\frac{\ket{00}+\ket{01}-\ket{10}-\ket{11}}{2},\frac{\ket{00}-\ket{01}-\ket{10}+\ket{11}}{2},\frac{\ket{00}-\ket{01}+\ket{10}-\ket{11}}{2}\right\}, \\
    B_2 &: \left\{\frac{\ket{00}-\ket{01}-i\ket{10}-i\ket{11}}{2},\frac{\ket{00}-\ket{01}+i\ket{10}+i\ket{11}}{2},\frac{\ket{00}+\ket{01}+i\ket{10}-i\ket{11}}{2},\frac{\ket{00}+\ket{01}-i\ket{10}+i\ket{11}}{2}\right\}, \\
    B_3 &: \left\{\frac{\ket{00}-i\ket{01}-i\ket{10}-\ket{11}}{2},\frac{\ket{00}-i\ket{01}+i\ket{10}+\ket{11}}{2},\frac{\ket{00}+i\ket{01}+i\ket{10}-\ket{11}}{2},\frac{\ket{00}+i\ket{01}-i\ket{10}+\ket{11}}{2}\right\}, \\
    B_4 &: \left\{\frac{\ket{00}-i\ket{01}-\ket{10}-i\ket{11}}{2},\frac{\ket{00}-i\ket{01}+\ket{10}+i\ket{11}}{2},\frac{\ket{00}+i\ket{01}-\ket{10}+i\ket{11}}{2},\frac{\ket{00}+i\ket{01}+\ket{10}-i\ket{11}}{2}\right\}. 
\end{align*}
\end{subequations}

In the procedure of the $(2,2,2)$ multi-copy strategy derived from the qudit verification strategy $\Omega_{\rm IV}$, both parties have an equal probability of initiating a test. If the test commences with party $A$, it holds a probability $p_0$ of measuring its two qubits on the first MUBs $B_0$. Additionally, party $A$ possesses a probability of $(1 - p_0)/4$ for measurement on the remaining $d^k$ mutually unbiased bases. Subsequently, party $A$ communicates its measurement choice and results $\ket{u_{li}}$ to party $B$, where $\ket{u_{li}}$ denotes the $i$-th basis of the $l$-th set of MUBs. Party $B$ then performs a measurement on the basis containing the reduced state 
$\ket{v_{li}} = \braket{u_{li}}{\Psi_\theta}/|\braket{u_{li}}{\Psi_\theta}|^2$ 
based on the message from Alice and passes the test if the result matches this reduced state.

Ref.~\cite{Li_2020} showed that setting $p_0 = (s_0^2 + s_1^2)/(2 + s_0^2 + s_1^2)$, where $s_0$ and $s_1$ represent the largest and second largest terms of the coefficient set $\{\cos^2\theta, \cos\theta \sin\theta , \sin^2\theta\}$, respectively, achieves the optimal strategy for verifying $\ket{\Psi_\theta}$. 
When $\theta \in (0, \pi/4)$, it holds that $s_0 = \cos^2 \theta$ and $s_1 = \sin \theta \cos \theta$. 
Consequently, the second largest eigenvalues of the verification strategy for $\ket{\Psi_\theta}$
is given by $\lambda_2(\Omega_{\rm IV})=\cos^2 \theta/(2+\cos^2 \theta)$.
Utilizing the result from the main text, we can conclude that for small values of $\varepsilon$, 
the number of copies required to achieve a certain worst-case failure probability $\delta$ is upper bounded by
\begin{align}\label{eq:dmext_N}
    N_{{\rm de},2}(\Omega_{\rm IV}) = \frac{2 + \cos^2 \theta}{2\varepsilon}\ln\frac{1}{\delta}.
\end{align}

\section{Comparation with related works}
\label{appx:comparision}

In this appendix, we compare our results with existing quantum-memory based quantum state verification studies and Bell sampling based works.
We first briefly summarize and then describe in detail the essential differences. 

\textit{Difference with quantum-memory based quantum state verification works.}
Briefly, compared to the Error Number Gates (ENG) strategy~\cite{miguel2022collective}, 
our two-copy strategy involves simpler Bell measurements and can be applied to a broader range of states. 
Unlike the non-demolition strategy~\cite{Liu_2021}, our task setting does not require transferring qubits between different verifiers.
This underscores the novelty and practicability of our multi-copy verification tasks. 
Compared to the stabilizer state verification strategies based on non-adaptive measurements~\cite{pallister2018optimal},
Our graph state verification strategy, for a $n$-qubit graph state, either (i) reduce the number of measurement settings from $2^n - 1$ to $1$ while realize a constant factor sample complexity improvement, or (ii) reduce the number of measurement settings from $n - 1$ to $1$ while improve the sample complexity from $\mathcal{O}(n)$ to $\mathcal{O}(1)$.

\textit{Difference with Bell sampling works.}
The idea of sampling in the Bell basis to learn about the properties of quantum states already has found many applications in quantum computing, 
including but not limited to learning and testing stabilizer states~\cite{montanaro2017learning,Gross-2021-Commun.Math.Phys.}, measuring magic~\cite{haug2023scalable},
and Bell sampling as a universal model of quantum computation~\cite{PhysRevLett.133.020601}.
Our two-copy verification strategy is similar to Bell sampling protocols. 
The novelty of our approach lies in being the first to demonstrate that Bell sampling of stabilizer states is 
actually a globally optimal strategy in the quantum state verification task setting.

In constructing the Bell measurement strategy, as mentioned in Appendices~\ref{app:channel},~\ref{app:proofeqGraph} and~\ref{app:graph_disentangled}, 
we first regard the quantum verification protocol as finding a channel that maximally preserves 
the distinguishability between two states. Then, we use the one-bit teleportation protocol to 
construct such an information-preserving channel. 
This scheme is illuminating and might provide a better understanding of why Bell measurements behave better.

\subsection{Comparision with ENG operations based quantum state verification}

In the first strategy proposed by Miguel-Ramiro et al.~\cite{miguel2022collective}, Error Number Gates (ENG) are used to encode noise from noisy Bell states into a \(d\)-dimensional maximally entangled state. Specifically, to collectively certify two noisy Bell states, they require ancilla qudits \(A_a\) and \(B_a\) in the state \(\ket{\Phi} = \frac{1}{2}(\ket{00}+\ket{11}+\ket{22}+\ket{33})_{A_a B_a}\). For each noisy Bell state, for example, the first one (\(A_1, B_1\)), the ENG operation is \(CX_{A_1 \to A_a} \otimes CX_{B_1 \to B_a}\), where \(CX_{C \to T}\) is a hyper-control-\(X\) gate that realizes \(\ket{u}_T = \ket{u+1 \mod 4}\) if the control qubit is in the state \(\ket{1}_C\). After the ENG encoding, the ancilla qudits are properly measured to complete the certification.

The first significant difference is that our strategy does not require perfectly entangled states. 
In our two-copy verification strategy, the verifiers need to conduct two-qubit Bell measurement, 
which is a different in experimental settings.

In~\cite{miguel2022collective}, they also mentioned the generalization that the \(d\)-dimensional maximally entangled state is not necessary 
because it can always be obtained by directly embedding noisy copies of the initial ensemble. 
For this generalized case, the difference between our work and theirs is summarized as follows.

For our two-copy verification strategy for graph states, which certainly covers the noisy Bell state, only qubit control-\(X\) gates between locally restored qubits are required to implement the transversal Bell measurements. However, the ENG operation for the embedding case requires Toffoli gates to equivalently realize the hyper control-\(X\) gate. In experiments, Toffoli gates are more complex compared to \(CNOT\) gates. This concludes that our two-copy strategy is simpler and can be applied to more general quantum states compared to the ENG operations in~\cite{miguel2022collective}.
Taking the two-copy case as an example, the ancilla state \(\ket{\Phi}\) defined earlier now becomes two embedded Bell states \(Q_0Q_1\) and \(P_0P_1\) through the equivalence \(\ket{2m+n}_{A_a} \to \ket{mn}_{Q_1Q_0}\).

\begin{align}
    \ket{\Phi} = \frac{1}{2}(\ket{00}_{Q_1Q_0}\ket{00}_{P_1P_0}+\ket{01}_{Q_1Q_0}\ket{01}_{P_1P_0}+\ket{10}_{Q_1Q_0}\ket{10}_{P_1P_0}+\ket{11}_{Q_1Q_0}\ket{11}_{P_1P_0}).
\end{align}
Similarly, in the embedding case, the gate \(CX_{C \to T_1T_0}\) realizes the linear transformation below if the control qubit \(C\) is in the state \(\ket{1}\):
\begin{align}
   \ket{00}_{T_1T_0} \to \ket{01}_{T_1T_0},\ket{01}_{T_1T_0} \to \ket{10}_{T_1T_0},\ket{10}_{T_1T_0}\to \ket{11}_{T_1T_0},\ket{11}_{T_1T_0} \to \ket{00}_{T_1T_0},
\end{align}
Writing such a linear transformation using qubit gates actually gives \(CX_{C;T_0} \circ CCX_{C,T_0;T_1}\), as shown in Fig.~\ref{fig:comp_1}.

\begin{figure}[!htbp]
    \centering
    \begin{quantikz}[row sep=0.3cm, column sep=0.3cm]
  \lstick{$C$}             \qw     &\ctrl{1}  &\qw    \\
  \lstick{$T(d = 4)$} \qw & \gate{X_4}  &\qw    
\end{quantikz}
\hspace{1cm}
    \begin{quantikz}[row sep=0.3cm, column sep=0.3cm]
  \lstick{$C$} & \ctrl{2}    &\ctrl{1}  &\qw    \\
  \lstick{$T_0$} & \control{} &  \targ{} &\qw \\
  \lstick{$T_1$} & \targ{} &  \qw    & \qw     
\end{quantikz}
    \caption{ In the embedding setting, the equivalent qubit gates implement the hyper control-\(X\) gate with \(d = 4\).}
    \label{fig:comp_1}
\end{figure}
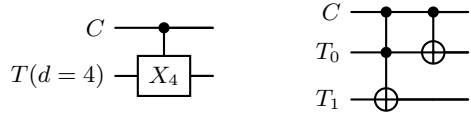

In~\cite{miguel2022collective}, it was also stated that embedding without ENG operations can approach the global optimal strategy. 
They used \(m\)-copies of noisy Bell states and directly measure the amplitude index. 
This is similar to our dimension expansion method, which treats the multi-copy qubit state as a qudit state 
and then applies a certain efficient local verification strategy for qudits.

\subsection{Comparision with non-demolition measurements based quantum state verification}

In the non-demolition measurements based quantum state verification originally proposed by Liu \emph{et al.}~\cite{Liu_2021}, 
the verifier first entangles the noisy state with an ancilla system, followed by a measurement on the ancilla. 
The measurement of ancilla qubits does not destroy the noisy copy, allowing them to be used in subsequent tests.

Taking the noisy Bell state as an example, an ancilla qubit in state \(\ket{0}\) is prepared and then interacts with the noisy Bell state according to the first circuit in Figure \ref{fig:comp_2}. After this interaction, we measure the ancilla qubit in the \(Z\) basis. If the result is \(\ket{0}\), we use a fresh ancilla qubit again and measure it in the \(Z\) basis after implementing the second circuit. The verification will pass if both results are \(\ket{0}\). As shown in reference~\cite{Liu_2021}, this strategy achieves global optimality.

\begin{figure}[!htbp]
    \centering
    \begin{quantikz}[row sep=0.3cm, column sep=0.3cm]
  \lstick[2]{noisy state}&\ctrl{2}  &\qw      & \qw \\
                         & \qw      &\ctrl{1} & \qw \\
  \lstick{$\ket{0}$}     &\targ{}   &\targ{}  & \qw
\end{quantikz}
\hspace{1cm}
       \begin{quantikz}[row sep=0.3cm, column sep=0.3cm]
  \lstick[2]{noisy state}&\gate{H} &\ctrl{2}  &\qw      &\gate{H} & \qw  \\
                         &\gate{H} & \qw      &\ctrl{1} &\gate{H} & \qw \\
  \lstick{$\ket{0}$}     &\qw      &\targ{}   &\targ{}  &\qw      & \qw
\end{quantikz}
    \caption{Circuits for the experimental realization of  
    quantum verification of noisy Bell states using non-demolition measurements. 
    The target Bell state is input on the first two qubits, and $\ket{0}$ represents the ancilla qubit.}
    \label{fig:comp_2}
\end{figure}
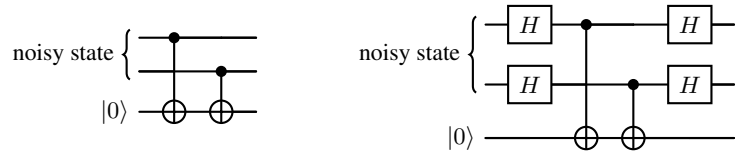

To compare with our memory-assisted scheme, we note that the ancilla qubits should interact with every qubit in the noisy state, as shown in Figure \ref{fig:comp_2}. This task setting differs from our assumption that verifiers can only implement local gates and exchange classical information. From an experimental point of view, it is also more resource-intensive to transfer ancilla qubits among multiple verifiers.

\subsection{Comparision with non-adaptive measurements based stabilizer state verification}

It is well-known that stabilizer states can be verified efficiently non-adaptively, 
as shown in the Theorems 6 and 7 in the work by Pallister \emph{et al.}~\cite{pallister2018optimal}. 
To justify the improvement through using Bell measurements, we noted that there are two critical factors to adjust the efficiency of a verification strategy. One, as mentioned by you, is the sample complexity $N(\Omega)$, and the other is the number of measurement settings. In fact, research~\cite{PhysRevA.104.062439} has shown the minimum number of measurement settings required to 
non-adaptively verify bipartite pure states.

For the first stabilizer verification strategy by Pallister \emph{et al.}~\cite[Theorem 6]{pallister2018optimal},
if $n$ is the number of qubits,
one randomly chooses one stabilizer from the $2^n - 1$ stabilizers of the stabilizer state $\ket{\psi}$ and projects onto the positive eigenspace of this stabilizer. This strategy has an efficiency of ${(2^n - 1)}/{2^{n - 1}} \times \varepsilon^{-1} \ln {\delta}^{-1}$. Indeed, our strategy only has a constant sample complexity improvement because ${(2^n - 1)}/{2^{n - 1}} < 2$ for all $n$. However, this strategy requires $2^n - 1$ measurement settings, which becomes particularly challenging when it is difficult or slow to switch measurement settings, as is the case in many practical scenarios. In our strategy, only one measurement setting is required.

The second strategy by Pallister \emph{et al.}~\cite[Theorem 7]{pallister2018optimal} improves the number of measurement settings. 
It chooses one stabilizer from $n$ stabilizer generators of $\ket{\psi}$ and projects onto the positive eigenspace of this generator.
This strategy requires only $n$ measurement settings.
However, the sample complexity now becomes $N = n \varepsilon^{-1} \ln \delta^{-1}$.
This complexity scales linearly with the number of qubits $n$, thus concluding a $\mathcal{O}(n)$ improvement in sample complexity for our strategy.

To the best of our knowledge, there is still no strategy that can use constant measurement settings and constant sample complexity to verify stabilizer states. We thus expect the existence of quantum memory or collective measurements can help to find a strategy with both better sampling complexity and fewer measurement settings.

\subsection{Comparision with quantum state verification in the adversarial scenario}

In the standard quantum state verification~\cite{pallister2018optimal}, 
the adversary can only prepare \emph{independent} multipartite entangled states and send to the local verifiers.
After receiving a quantum state, the verifiers conduct local measurments and make decision based on the measurement outcomes.
From the quantum memory perspective, neither the adversary nor the verifiers have access to quantum memory.

In quantum state verification in the adversarial scenario~\cite{Zhu_2019_adver},
the adversary is much more powerful and can produce an arbitrarily correlated or even entangled state
among many state copies, which is applicable to the case of \emph{nonindependent} sources. 
After receiving the multipartite quantum state in many copies, the verifiers conduct local measurments and make decision based on the measurement outcomes.
Since the verifiers have no quantum memory, they have to conduct measurements qudit by qudit
but cannot conduct measurements across many qudits.
From the quantum memory perspective, the adversary possesses a quantum memory but the verifiers do not have quantum memory.

Our work generalizes the standard quantum state verification and is converse to the quantum state verification in the adversarial scenario:
In the verification task under our consideration, the adversary does not possess quantum memory but the verifiers have access to quantum memory.
Intuitively, we can expect that the verifiers can acomplish the verification tasks more efficiently since they
can store the quantum states and measure them collectively to make decision.
The main contribution of our work is that we rigorously justify this intuition by establishing 
a mathematical framework and proposing various strategy construction techniques to accomplish this task.

In the most general case, we can consider quantum state verification where both the adversary and 
the verifiers have access to different amount of quantum memory, quantified by the number of qudits that they can 
store before processing. This corresponds to quantum state verification in the experimental and practical situations.
We leave this interesting problem as future work.

\end{document}